\documentclass{lmcs}
\pdfoutput=1

\usepackage{lastpage}
\lmcsdoi{17}{3}{7}
\lmcsheading{}{\pageref{LastPage}}{}{}%
{Jan.~20,~2020}{Jul.~21,~2021}{}

\usepackage[utf8]{inputenc}

\usepackage{microtype}

\usepackage{amsmath}
\usepackage{amsthm}
\usepackage{wrapfig}
\usepackage{amssymb}
\usepackage{stmaryrd}
\usepackage{thmtools}
\usepackage{epsfig}
\usepackage{bbold}
\usepackage[T1]{fontenc}
\usepackage{multicol}
\usepackage{float}

\usepackage{graphicx}
\usepackage{sidecap}

\usepackage[textwidth=2cm,textsize=small]{todonotes}

\usepackage{pgf}
\usepackage{tikz}
\usetikzlibrary{arrows,decorations.pathmorphing,backgrounds,positioning,fit,calc,automata}
\usetikzlibrary{trees}
\usetikzlibrary{shapes}

\usepackage{multirow}

\newenvironment{claimproof}[1]{\par\noindent{\it Proof.}\space#1}{\hfill $\blacksquare$}

\newcommand{\NLogSpace}{\textsc{NLogSpace}}

\newcommand{\AL}{\Sigma}
\newcommand{\sem}[1]{{\llbracket{}{#1}\rrbracket}}

\newcommand{\SR}{\mathbb{S}}
\newcommand{\MD}{\mathbb{M}}

\newcommand{\sr}{S}
\newcommand{\add}{\oplus}
\newcommand{\mult}{\odot}

\newcommand{\zero}{\mathbb{0}}
\newcommand{\one}{\mathbb{1}}

\newcommand{\nat}{\mathbb{N}}
\newcommand{\natinf}{\mathbb{N} \cup \{\infty\}}
\newcommand{\natninf}{\mathbb{N} \cup \{-\infty\}}
\newcommand{\natpinf}{\mathbb{N} \cup \{+\infty\}}

\newcommand{\trop}{\nat_{\min,+}}

\newcommand{\arctic}{\nat_{\max,+}}
\newcommand{\natsemiring}{\nat_{+, \times}}
\newcommand{\bool}{\mathbb{B}}

\newcommand{\cA}{\mathcal{A}}

\newcommand{\trans}[2][]{\raisebox{-1pt}[10pt][0pt]{$\overset{#2}{\underset{^{#1}}{\raisebox{0pt}[3pt][0pt]{$\relbar\mspace{-8mu}\rightarrow$}}}$}}

\newcommand{\Run}{\operatorname{Run}}

\newcommand{\asem}[1]{{\sem{#1}}}

\tikzset{
defaultstyle/.style={>=stealth,semithick, auto,font=\small,
initial text= {},
initial distance= {3.5mm},
accepting distance= {3.5mm}},
accepting/.style=accepting by arrow,
nstate/.style={circle, semithick,inner sep=1pt, minimum size=4mm}}

\newcommand{\proper}{linear}
\newcommand{\compatible}{decomposable}
\newcommand{\wt}{\operatorname{wt}}

\begin{document}
\title{Pumping lemmas for weighted automata}
\titlecomment{{\lsuper*}A preliminary version of this paper appeared in the proceedings of the 35th Symposium on Theoretical Aspects of Computer Science (STACS)~\cite{MazowieckiR18}.}

\author[A.~Chattopadhyay]{Agnishom Chattopadhyay}
\address{Chennai Mathematical Institute, India}
\email{agnishom@cmi.ac.in}
\thanks{Agnishom Chattopadhyay was an intern at the University of
  Bordeaux, partially supported by the UMI ReLaX and the French-Indian
  project IoTTTA}

\author[F.~Mazowiecki]{Filip Mazowiecki}
\address{Max Planck Institute for Software Systems, Germany}
\email{filipm@mpi-sws.org}
\thanks{Filip Mazowiecki was partially supported by EPSRC grants
  EP/M011801/1 and EP/M012298/1. The work was partly done during a postdoctoral stay
  at the University of Bordeaux.}

\author[A.~Muscholl]{Anca Muscholl}
\address{Universit\'{e} de Bordeaux, France}
\email{anca@labri.fr}
\thanks{Anca Muscholl was partially supported by the ANR project DeLTA (ANR-16-CE40-0007).}

\author[C.~Riveros]{Cristian Riveros}
\address{Pontificia Universidad Cat\'olica de Chile, Santiago, Chile}
\email{cristian.riveros@uc.cl}
\thanks{Cristian Riveros was funded by ANID -- Millennium Science Initiative
Program -- Code ICN17\_002}

\begin{abstract}
We present pumping lemmas for five classes of functions definable by fragments of weighted automata over the min-plus semiring, the max-plus semiring and the semiring of natural numbers. As a corollary we show that the hierarchy of functions definable by unambiguous, finitely-ambiguous, polynomially-ambiguous weighted automata, and the full class of weighted automata is strict for the min-plus and max-plus semirings.
 \end{abstract}

\maketitle

\section{Introduction}

Weighted automata (WA for short) are a quantitative extension of finite state
automata used to compute functions over words.
They have been extensively studied since Sch\"{u}tzenberger's early
works~\cite{Schuetzenberger56}, see also~\cite{DrosteHWA09}. In particular, decidability
questions~\cite{Krob92,AlmagorBK11}, model
extensions~\cite{DrosteG07}, logical
characterisations~\cite{DrosteG07,kreutzer2013quantitative}, and
various applications~\cite{Mohri97,culik1993image} have been thoroughly
investigated in recent years.

The class of functions computed by WA enjoys several equivalent
representations in terms of automata and logics. Alur et
al.~introduced the model of cost register automata (CRA for
short)~\cite{alur2013regular,alur2013decision}, an alternative model
to compute functions over words inspired by programming paradigms, that recently received a lot of attention~\cite{kickasspaper,MazowieckiR19,DaviaudRT16,AlmagorCMP18}.
It was proved that the linear fragment of CRA is
equally expressive to WA.
Regarding logics, Droste and Gastin introduced in~\cite{DrosteG07} the so-called Weighted Logic (WL), a natural extension of monadic second order logics (MSO) from the boolean semiring to a commutative semiring.
It was shown in~\cite{DrosteG07} that a natural syntactic restriction of WL is equally expressive to WA, giving the first logical characterisation of~WA (see also~\cite{kreutzer2013quantitative}).

The decidability and complexity of various decision problems for WA
have also been investigated, unfortunately often with negative
results, such as for the equivalence
problem~\cite{Krob92,AlmagorBK11}. Another basic question, the decidability of the
determinisation problem for WA, is still open.
For this reason, research has focused on various fragments of WA over different semirings.
For example, over a one-letter alphabet, where WA are equivalent to
linear recurrences, some new decidability results were recently shown for limited fragments~\cite{OuaknineW14,OuaknineW14a,AkshayBMV020}.
Further restrictions of WA involve bounding their numbers of
runs. Among them, the  most studied classes are \emph{unambiguous automata}, \emph{finitely-ambiguous automata}, and \emph{polynomially-ambiguous automata}, where the numbers of accepting runs are bounded by one, a  constant, a polynomial in the size of input, respectively~\cite{Weber94,KlimannLMP04,KirstenL09}.
These subclasses of WA turned out to be robust, enjoying equivalent characterisations in terms of cost register automata~\cite{alur2013regular} and weighted logics~\cite{kreutzer2013quantitative}.

Although functions defined by WA and subclasses thereof have been studied
in terms of representations and decidability, not much is known about
expressibility issues.
Indeed, we are not aware of any general techniques to show if a
function is definable or not by some WA, or any of its subclasses.
Results related to the expressiveness of  WA usually require sophisticated arguments for each particular function~\cite{KlimannLMP04, kickasspaper} and there is no clear path to generalise these techniques.
For some semirings strict inclusions between unambiguous, finitely-ambiguous, polynomially-ambiguous, and the full class of WA are ``well-known'' to the community, but it is hard to find references to formal proofs (see related work below).
In contrast, for regular languages or first-order logics there exist
elegant and powerful techniques for showing inexpressibility, as for
example, the Myhill-Nerode congruence for regular
languages~\cite{HopcroftU79} or Ehrenfeucht-Fra\"{\i}ss\'{e} games~\cite{fraise1984quelques,ehrenfeucht1961application,libkin2013elements}
and aperiodic congruences for first-order logic~\cite{sch65}.
One would like to have similar techniques in the quantitative world
that simplify inexpressibility arguments for WA, CRA or WL.
Such techniques would help to understand the inner structure of these functions and unveil their limits of expressibility.

In this paper, we embark in the work of building an expressibility toolbox for weighted automata. We present five pumping lemmas, each of them for a  different class or subclass of functions defined by WA over the min-plus semiring, the max-plus semiring or the semiring of natural numbers. For each pumping lemma we show examples of functions that do not satisfy the lemma, giving very short inexpressibility proofs.
Our results do not attempt to fully characterise the class or subclasses of weighted automata in terms of pumping properties, nor to provide conditions that can be verified by a computer.
Our goal is to devise a systematic way to reason about expressibility of weighted automata and to provide simple arguments to show that functions do not belong to a given class.

\emph{Related work.} In~\cite{Kirsten08}, it is shown that over the
min-plus semiring polynomially-ambiguous automata are strictly more
expressive than finitely-ambiguous automata. In~\cite{KlimannLMP04}
strict inclusions between unambiguous automata, finitely-ambiguous
automata, and the full class of WA are shown over the max-plus
semiring. Using results from~\cite{ColcombetKMT16} one can
deduce that unambiguous automata are strictly included in the other
classes over the min-plus and max-plus semirings. More recently, \cite{DrosteG19} studied aperiodic WA over
arbitrary weights, relating fragments of aperiodic WA with various
degrees of ambiguity, and providing separating examples over the
min-plus, max-plus and the natural semiring. In all these papers the strict inclusions are shown by analysing
particular functions.  Gathering these
results we obtain strict inclusions between unambiguous automata,
finitely-ambiguous automata, and the full class of WA over the
min-plus semiring. However, to our knowledge, there is no reference
for a strict inclusion between polynomially-ambiguous automata and the
full class of WA. There is some work on the semiring of rational numbers with the usual
sum and product~\cite{MazowieckiR19,Barloy2019}. In these papers the polynomially-ambiguous fragment over the one-letter alphabet is characterised in terms of a fragment of
linear recurrence sequences. Both papers provide proofs that
polynomially-ambiguous weighted automata are strictly contained in the full class
of weighted automata over the semiring of rationals.

\emph{Differences with the conference version.} Compared
to~\cite{MazowieckiR18}, we present new pumping lemmas for the max-plus
semiring (Section~\ref{sec:max}) regarding finitely ambiguous and polynomially ambiguous max-plus automata. As a corollary we obtain a strict hierarchy of functions similar to the one for min-plus automata.

\emph{Organization}. In Section~\ref{sec:preliminaries} we introduce
weighted automata and some basic definitions. In
Section~\ref{sec:pump-sum} and Section~\ref{sec:fin_min} we present
pumping lemmas for weighted automata over the semiring of
natural numbers and its extension using the operation $\min$. In
Section~\ref{sec:pumping} we show the pumping lemma for
polynomially-ambiguous automata over the min-plus semiring, then we
turn to the max-plus semiring in Section~\ref{sec:max}. Concluding remarks can be found in Section~\ref{sec:conclusions}.

\section{Preliminaries}
\label{sec:preliminaries}
In this section, we recall the definitions of weighted automata. We start with the definitions that are standard in this area.
A monoid $\MD = (M,\otimes,\one)$ is a set $M$ with an associative operation $\otimes$ and a neutral element $\one$. Standard examples of monoids are: the set of words $\Sigma^*$ with concatenation and empty word; or the set of matrices with multiplication and the identity matrix.
A semiring is a structure $\SR = (\sr, \add, \mult, \zero, \one)$, where $(\sr, \add, \zero)$ is a commutative monoid, $(\sr, \mult, \one)$ is a monoid, multiplication distributes over addition, and $\zero \mult s = s \mult \zero = \zero$ for each $s \in \sr$.
If the multiplication is commutative, we say that $\SR$ is commutative. In this paper, we always assume that $\SR$ is commutative.
We usually denote $S$ or $M$ by the name of the semiring or monoid $\SR$ or $\MD$.
We are interested mostly in
the tropical semirings: the \emph{min-plus semiring} $(\natpinf, \min, +,+\infty, 0)$ and the \emph{max-plus semiring} $(\natninf, \max, +, -\infty, 0)$. We are also interested in the \emph{semiring of natural numbers with infinity} $(\natinf, +, \cdot, 0, 1)$, where $\infty + n = \infty$ for every $n \in \natinf$ and $\infty \cdot n = \infty$ if $n \neq 0$ and $0$ otherwise.
We denote the tropical semirings by $\trop$ and $\arctic$; and the latter semiring by $\natsemiring$.
Note that $\natsemiring$ is an extension of the standard semiring of natural numbers $\nat$ and all our results for $\natsemiring$ also hold for $\nat$.
We use the extended version of $\nat$ to transfer some results from $\natsemiring$ to $\trop$ and $\arctic$ (see Section~\ref{sec:pump-sum} and Section~\ref{sec:fin_min}). Notice that in $\natsemiring$ we did not put a sign in front of $\infty$. For the semiring structure, this is not relevant. However, some statements in Section~\ref{sec:pump-sum} and Section~\ref{sec:fin_min} will assume that the semiring is given with an order. Then one should think that the results hold both when $\infty = + \infty$ and when $\infty = - \infty$.

Given a finite set $Q$, we denote by $\SR^{Q \times Q}$ ($\SR^Q$) the set of square matrices (vectors resp.) over $\SR$ indexed by $Q$. The algebra induced by $\SR$ over $\SR^{Q \times Q}$ and $\SR^Q$ is defined as usual.

We also consider two finite semirings that will be useful during proofs. We consider the boolean semiring $\bool = (\{0,1\}, \vee, \wedge, 0, 1)$ and the extended boolean semiring $\bool_\infty = (\{0,1, \infty\}, \vee, \wedge, 0, 1)$ such that $\infty \vee n = \infty$ for every $n \in \{0,1, \infty\}$, $\infty \wedge 0 = 0$, and $\infty \wedge n = \infty$ if $n \in \{1, \infty\}$.
These finite semirings will be used as \emph{abstractions} of $\trop$, $\arctic$ and $\natsemiring$.


\subsection{Weighted automata}
Fix a finite alphabet $\AL$ and a commutative semiring $\SR$.
A \emph{weighted automaton} (WA for short) over $\AL$ and $\SR$ is a tuple $\cA = (Q, \AL, \{M_a\}_{a \in \Sigma}, I, F)$ where $Q$ is a finite set of states, $\{M_a\}_{a \in \Sigma}$ is a set of matrices such that $M_a \in \SR^{Q \times Q}$ and $I, F \in \SR^Q$ are the initial and the final vectors, respectively~\cite{Sakarovitch09,DrosteHWA09}.
We say that a state $q$ is initial if $I(q) \neq \zero$ and accepting if $F(q) \neq \zero$.
We say that $(p,a,s,q)$ is a transition, where $M_a(p,q) = s$ and we write $p \:\trans{a/s}\: q$.
Furthermore, we say that a run $\rho$ of $\cA$ over a word $w = a_1 \ldots a_n$ is a sequence of transitions:
$
\rho = q_0 \:\trans{a_1/s_1}\: q_1  \:\trans{a_2/s_2}\: \cdots\:\trans{a_n/s_n}\: q_n,
$
where $s_i \neq \zero$ for all $1 \le i \le n$ and $I(q_0) \neq \zero$.
We refer to $q_i$ as the $i$-th state of the run $\rho$.
The run $\rho$ is accepting if $F(q_n) \neq \zero$, and the weight of an accepting run $\rho$ is defined by
$
|\rho| = I(q_0) \mult (\bigodot_{i=1}^{n} s_i) \mult F(q_n).
$
We define $\Run_\cA(w)$ as the set of all accepting runs of $\cA$ over~$w$.
Finally, the output of $\cA$ over a word $w$ is defined by
$
\asem{\cA}(w)  =  I^t \cdot M_{a_1} \cdot \ldots \cdot M_{a_n} \cdot F  = \bigoplus_{\rho \in \Run_\cA(w)} |\rho|
$
where $I^t$ is the transpose of $I$ and the sum is equal to $\zero$ if $\Run_\cA(w)$ is empty. For a word $w = a_1 \ldots a_n$, by $M_w$ we denote $M_{a_1} \cdot \ldots \cdot M_{a_n}$, so that $\asem{\cA}(w) = I^t \cdot M_w \cdot F$.
Note that $M_w(p,q)$ provides the cost of moving from state $p$ to
state $q$ reading the word $w$. Functions defined by weighted automata are called \emph{recognisable functions}.

In this paper, we study the specification of functions from words to values, namely, from $\Sigma^*$ to $\SR$.
We say that a function $f: \Sigma^* \rightarrow \SR$ is definable by a weighted automaton $\cA$ if $f(w) = \asem{\cA}(w)$ for all $w \in \Sigma^*$.

A weighted automaton $\cA$ is called \emph{unambiguous} (U-WA) if $|\Run_\cA(w)| \leq 1$ for every $w \in \Sigma^*$; and $\cA$ is called \emph{finitely-ambiguous} (FA-WA) if there exists a uniform bound $N$ such that $|\Run_\cA(w)| \leq N$ for every $w \in \Sigma^*$~\cite{Weber94,KlimannLMP04}.
Furthermore, $\cA$ is called \emph{polynomially-ambiguous} (PA-WA) if the function $|\Run_\cA(w)|$ is bounded by a polynomial in the length of~$w$~\cite{KirstenL09}.
We call the classes of functions definable by such automata \emph{unambiguously recognisable}, \emph{finite-ambiguously recognisable} and \emph{polynomial-ambiguously recognisable}.

Note that every unambiguous WA over $\trop$ and $\arctic$ can be defined by a
WA over the semiring
$\natsemiring$ (recall that
$\infty$ is in $\natsemiring$).
Indeed,
let $\cA$ be an unambiguous WA over $\trop$ (or $\arctic$). We define two copies of $\cA$: $\cA_1$ and
$\cA_2$ such that all previously non-$\zero$ transitions between the states inside each copy have weights $1$. The initial vector is inherited from $\cA$ in the copy $\cA_1$ and defined as $0$ for states in $\cA_2$. Conversely, the final vector is inherited from $\cA$ in the copy $\cA_2$ and defined as $0$ for states in $\cA_1$. Finally, for every transition $(p,a,s,q)$ in $\cA$ we add a transition $(p_1,a,s,q_2)$, where $p_1$ is the copy of $p$ in $\cA_1$ and $q_2$ is the copy of $q$ in $\cA_2$. Since $\cA$ is unambiguous it is easy to see that the construction works for words $w$ such that $\cA(w) \neq \infty$. Let $L$ be the language of the remaining words. Then for $w\in L$ the new automaton outputs $0$ instead of $\infty$. Since $L$ is regular it suffices to take a union with another weighted automaton over $\natsemiring$ that outputs $\infty$ for $w \in L$ and $0$ for $w \not \in L$.



Therefore, the class of unambiguously recognisable functions over $\trop$
(and $\arctic$) is included in the class of recognisable functions over
$\natsemiring$.
The inclusions are strict since recognisable functions over $\trop$ (and $\arctic$) are always bounded by a linear function in the size of the word, and it is easy to define the function $f(w) = 2^{|w|}$ over $\natsemiring$.
Below, we give several examples of functions defined by WA over $\natsemiring$ and $\trop$ that will be used in paper. Recall that in the latter semiring $\zero = \infty$ and $\mult = +$. Transitions $p \:\trans{a/s}\: q$, where $s = \zero$, are omitted.

\begin{figure}[]
	\begin{center}
		\scalebox{.75}{
			  	\begin{tikzpicture}[-\string>,\string>=stealth',shorten \string>=1pt,auto,node distance=2.5cm,semithick,initial text={}, every node/.style={scale=0.9}]
	\tikzstyle{every state}=[fill=white,draw=black,text=black]

	\node[state, double, minimum size=5ex] 		(n1) at (-0.5,0) {};
	\node[state, left = 1cm of n1, initial left, minimum size=5ex] 		(n0) {};
	\node (no1) at ($(n1) + (-0.8, -0.8)$) {};
	\draw (no1) edge (n1);

	\draw (n0) edge node {$b \ / \ 0$}  (n1);

	\draw (n0) edge[loop above] node {$\begin{array}{c}
		a \ / \ 0 \\
		b \ / \ 0
		\end{array}$}  (n0);

	\draw (n1) edge[loop above] node {$\begin{array}{c}
		a \ / \ 1
		\end{array}$}  (n1);

	\node at ($(n0) + (-0.3, -0.7)$) {$\mathcal{W}_1$ over $\trop$};

	\node[state, right = 3.5cm of n1, initial below, minimum size=5ex] 		(m0) {};
	\node[state, right = 1cm of m0, double, minimum size=5ex] 		(m1) {};
	\node[state, left = 1cm of m0, double, minimum size=5ex] 		(m2) {};
		\node[state, right = 0.7cm of m1, double, minimum size=5ex,initial below] 		(m3) {};
		\node[below right = 0.2cm and -0.3cm of m3] {$\infty$};

	\draw (m0) edge node {$a \ / \ 1$}  (m1);
	\draw (m0) edge node[above] {$b \ / \ 0$}  (m2);

	\draw (m0) edge[loop above] node {$\begin{array}{c}
		a \ / \ 1 \\
		b \ / \ 1
		\end{array}$}  (m0);

	\draw (m1) edge[loop above] node {$\begin{array}{c}
		a \ / \ 1
		\end{array}$}  (m1);

	\node at ($(m2) + (-0.3, -0.7)$) {$\mathcal{W}_1'$ over $\natsemiring$};

	\node[state, double, initial left, minimum size=5ex,right = 2.5cm of m3]	(q0) {};
	\node[state,double, initial left, right = 1cm of q0, minimum size=5ex] 		(q1) {};

	\node (qo1) at ($(q1) + (-0.8, -0.8)$) {};
	\node (qo2) at ($(q1) + (0.8, -0.8)$) {};

	\draw (q0) edge[loop above] node {$\begin{array}{c}
				a\ / \ 1 \\
				b \ / \ 0
				\end{array}$}  (p0);
	\draw (q1) edge[loop above] node {$\begin{array}{c}
		a\ / \ 0 \\
		b \ / \ 1
		\end{array}$}  (q1);

	\node at ($(q0) + (-0.3, -0.7)$) {$\mathcal{W}_2$ over $\trop$};

	\begin{scope}[yshift=-3cm]
	\node[state, minimum size=5ex] 		(p0) at (1.5,0) {};
	\node[state, left = 1cm of p0, initial left, minimum size=5ex] 		(p5) {};
	\node[state,double, right = 1cm of p0, minimum size=5ex] 		(p1) {};
	\node[state,double, right =1cm of p1, minimum size=5ex] 		(p2) {};

	\node (po1) at ($(p0) + (-0.8, -0.8)$) {};
	\node (po2) at ($(p1) + (0.8, -0.8)$) {};
	\draw (po1) edge (p0);

	\draw (p1) edge[loop above] node {$\begin{array}{c}
		b \ / \ 1
		\end{array}$}  (p1);
	\draw (p2) edge[loop above] node {$a,b \ / \ 0$}  (p2);
		\draw (p5) edge[loop above] node {$a,b \ / \ 0$}  (p5);

	\draw (p0) edge node {$b \ / \ 1$}  (p1);
	\draw (p1) edge node {$a \ / \ 0$}  (p2);
	\draw (p5) edge node {$a \ / \ 0$}  (p0);

	\node at ($(p5) + (-0.3, -0.7)$) {$\mathcal{W}_4$  over $\trop$};

		\node[state, double, initial left, minimum size=5ex] 		(q0) at (-4,0) {};
	\node[state,double, right = 1cm of q0, minimum size=5ex] 		(q1) {};

	\node (qo1) at ($(q1) + (-0.8, -0.8)$) {};
	\node (qo2) at ($(q1) + (0.8, -0.8)$) {};

	\draw (q0) edge[loop above] node {$\begin{array}{c}
				a\ / \ 1 \\
				b \ / \ 0
				\end{array}$}  (p0);
	\draw (q1) edge[loop above] node {$\begin{array}{c}
		a\ / \ 0 \\
		b \ / \ 1
		\end{array}$}  (q1);
	\draw (q0) edge node {$\begin{array}{c}
				a\ / \ 0 \\
				b \ / \ 0
				\end{array}$} (q1);


	\node at ($(q0) + (-0.3, -0.7)$) {$\mathcal{W}_3$ over $\trop$};

		\node[state, double, minimum size=5ex] 		(r0) at (7,0) {};
	\node[state,double, right of=r0, minimum size=5ex] 		(r1) {};
	\node[state,double, right of=r1, minimum size=5ex] 		(r2) {};

	\node (ro1) at ($(r1) + (0.0, -1)$) {};
	\node (ro2) at ($(r2) + (0.0, -1)$) {};
	\node (ro3) at ($(r0) + (-1, 0)$) {};
	\draw (ro1) edge (r1);

	\draw (r0) edge[loop above] node {$\begin{array}{c}
				a\ / \ 1 \\
				b \ / \ 0
				\end{array}$}  (r0);
	\draw (r1) edge[loop above] node {$\begin{array}{c}
		\# \ / \ 0
		\end{array}$}  (r1);
	\draw (r2) edge[loop above] node {$\begin{array}{c}
		a\ / \ 0 \\
		b \ / \ 1
		\end{array}$}  (r2);

	\draw (r0) edge[bend left] node {$\# \ / \ 0$}  (r1);
		\draw (r2) edge[bend right] node[above] {$\# \ / \ 0$}  (r1);
	\draw (r1) edge[bend right] node[below] {$\begin{array}{c}
		a\ / \ 0 \\
		b \ / \ 1
		\end{array}$}  (r2);
		\draw (r1) edge[bend left] node[below] {$\begin{array}{c}
		a\ / \ 1 \\
		b \ / \ 0
		\end{array}$}  (r0);

	\node at ($(r0) + (-0.5, -0.7)$) {$\mathcal{W}_5$ over $\trop$};

	\end{scope}

	\end{tikzpicture}

		}
		\caption{Examples of weighted automata. For WA over $\trop$ the initial and accepting states are labelled by 0 in the corresponding vector, and $\infty$ otherwise. Similarly, for WA over $\natsemiring$ the initial and accepting states are labelled by 1 in the corresponding vector, and $0$ otherwise; except for the initial state labeled $\infty$.}\label{fig:WA}
	\end{center}
\end{figure}
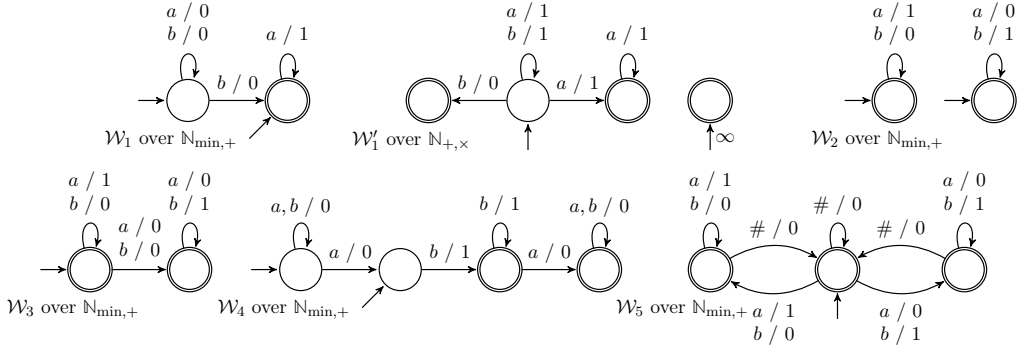

\begin{exa} \label{ex:in_unambiguous}
	Let $\Sigma = \{a,b\}$.
	Consider the function $f_1$ that for given word $w \in
        \Sigma^*$ outputs the length of the biggest suffix of $a$'s
        (and $\infty$ for the empty word). This is defined by the WA
        $\mathcal{W}_1$ over $\trop$ in Figure~\ref{fig:WA}. One can
        easily check that $\mathcal{W}_1$ is unambiguous, hence $f_1$
        is an unambiguously recognisable function over $\trop$. In
        Figure~\ref{fig:WA}, the WA $\mathcal{W}_1'$ over $\natsemiring$ also defines $f_1$.
\end{exa}

\begin{exa} \label{ex:min-letter}
	Let $\Sigma = \{a,b\}$.
	Consider the function $f_2$ that for a given word $w \in
        \Sigma^*$ outputs $\min\{|w|_a, |w|_b\}$, namely, counts the
        number of each letter and returns the minimum. This is defined
        by the WA $\mathcal{W}_2$ in Figure~\ref{fig:WA}. The WA
        $\mathcal{W}_2$ is finitely-ambiguous, hence $f_2$ is a finite-ambiguously recognisable function.
\end{exa}

\begin{exa} \label{ex:min-a_to_b}
	Let $\Sigma = \{a,b\}$.
	Consider the function $f_3$ that for a given word $w = a_1 \ldots a_n \in \Sigma^*$ outputs $\min_{0\le i \le n} \{|a_1 \ldots a_i|_a + |a_{i+1} \ldots a_n|_b\}$.
	This is defined by the WA $\mathcal{W}_3$ in Figure~\ref{fig:WA}. This
        WA is polynomially-ambiguous, hence $f_3$ is a
        polynomial-ambiguously recognisable function.
\end{exa}

\begin{exa} \label{ex:min-b-substrings}
	Let $\Sigma = \{a,b\}$.
	Consider the function $f_4$ that for a given word $w \in \Sigma^*$ computes the shortest subword of $b$'s (if there is none it outputs $+\infty)$.
	This is defined by $\mathcal{W}_4$ in Figure~\ref{fig:WA}. The WA is polynomially-ambiguous, hence $f_4$ belongs to polynomially-ambiguous functions.
\end{exa}

\begin{exa}	\label{ex:no-poly-ambiguous}
	Let $\Sigma = \{a,b,\#\}$.
	Consider the function $f_5$ such that, for every $w \in
        \Sigma^*$ of the form $w_0 \# w_1 \# \ldots \# w_n$ with $w_i
        \in \{a,b\}^*$, it computes $\min\{|w_i|_a, |w_i|_b\}$ for each subword $w_i$ and then it sums these values over all subwords $w_i$, that is, $f_5(w) = \sum_{i=0}^n \min\{|w_i|_a, |w_i|_b\}$.
	This function is defined by the WA $\mathcal{W}_5$ in Figure~\ref{fig:WA}. Given that this WA has an exponential number of runs, the function $f_5$ is a recognisable function, but not necessarily a polynomial-ambiguously recognisable function.
\end{exa}

We will also discuss variants of the functions $f_1$, $f_2$, $f_3$, $f_4$ and $f_5$ in the max-plus semiring in Section~\ref{sec:max}.

We assume that our weighted automata are always trim, namely, all
their states are reachable from some initial state (accessible in
short) and they can reach some final state (co-accessible in short).
Verifying if a state is accessible or co-accessible is reduced to a reachability test in the transition graph~\cite{papadimitriou1993computational} and this can be done in \NLogSpace.
Thus, we can assume without loss of generality that all our automata are trimmed.

\subsection{Finite monoids and idempotents} \label{subsec:monoids}
We say that a monoid is finite if the set of its elements is finite. Let $\MD = (M,\otimes,\one)$ be a finite monoid. We say that $\iota \in \MD$ is an idempotent if $\iota \otimes \iota = \iota$.
The following lemma is a standard result for finite monoids and
idempotents (see e.g.~\cite{pin2010mathematical}).
\begin{lem}\label{lemma:ramsey}
	Let $\MD$ be a finite monoid.
	There exists some $N > 0$ such that every sequence
	$m_1, \ldots, m_n$,
	with $m_i \in \MD$ and $n \geq N$, can be factorised as:
	\[
	(m_1 \otimes \ldots \otimes m_i) \otimes (m_{i+1} \otimes
        \ldots \otimes m_j) \otimes (m_{j+1} \ldots \otimes m_{n})\, ,
	\]
	where $i < j \leq n$ and $(m_{i+1} \otimes \ldots \otimes m_j)$ is an idempotent.
\end{lem}
We will mainly use two finite monoids of matrices, $\bool^{Q \times Q}$ and $\bool_\infty^{Q \times Q}$.
We define abstractions, i.e.\ homomorphisms of $\trop^{Q
  \times Q}$ to $\bool^{Q \times Q}$, $\arctic^{Q \times Q}$ to
$\bool^{Q \times Q}$, and $\natsemiring^{Q \times Q}$ to
$\bool_\infty^{Q \times Q}$. They are obtained from the homomorphisms
defined on elements of the matrices, namely $h_1 : \trop \to \bool$,
$h_2 : \arctic \to \bool$, and $h_3 : \natsemiring \to \bool_\infty$,
by setting
$h_1(m)= 0$ iff $m = + \infty$; $h_2(m)= 0$ iff $m = -\infty$; $h_3(0)=
0$, $h_3(\infty) = \infty$, and $h_3(m) = 1$ if $m \not=0,\infty$.
For matrices $M_1 \in \trop^{Q \times Q}$, $M_2 \in \arctic^{Q \times Q}$, or $M_3 \in \natsemiring^{Q \times Q}$ we denote by $\bar{M_1} = h_1(M_1)$, $\bar{M_2} = h_2(M_2)$, or $\bar{M_3} = h_3(M_3)$ their abstractions in $\bool^{Q \times Q}$ or $\bool_\infty^{Q \times Q}$.
By abuse of language we say that a matrix $M$ from $ \trop^{Q \times
  Q}$, $\arctic^{Q \times Q}$ or $\natsemiring^{Q \times Q}$ is
\emph{idempotent}, if its abstraction $\bar{M}$ is idempotent.

\section{Recognisable functions over $\natsemiring$}
\label{sec:pump-sum}

In this section we consider recognisable functions over $\natsemiring$. As
a corollary of the pumping lemma we show that FA-WA are strictly more
expressive than U-WA over $\trop$ and $\arctic$
(Example~\ref{ex:unambiguous} and beginning of Section~\ref{sec:max}).
Moreover, this shows that there are finite-ambiguously recognisable functions over $\trop$ and $\arctic$ that cannot be defined by any recognisable function over $\natsemiring$.

We introduce some notation to simplify the presentation.
Given $u \cdot v \cdot w = \hat{u} \cdot \hat{v} \cdot \hat{w}$, where $u,v,w, \hat{u}, \hat{v}, \hat{w} \in \Sigma^*$, we say that $\hat{u} \cdot \underline{\hat{v}} \cdot \hat{w}$ is a \emph{refinement} of $u \cdot \underline{v} \cdot w$ if there exist $u', w'$ such that $u \cdot u' = \hat{u}$, $w' \cdot w = \hat{w}$, $u' \cdot \hat{v} \cdot w' = v$, and $\hat{v} \neq \epsilon$. We underline the infixes $v$ and $\hat{v}$ to emphasise the refined part.

\begin{thm}[Pumping Lemma for recognisable functions over $\natsemiring$]
\label{theorem:pumping_refined}
Let $f : \Sigma^* \to \natinf$ be a recognisable function over $\natsemiring$.
There exists $N$ such that for all words of the form $u \cdot v \cdot
w \in \Sigma^*$ with $|v| \geq N$, $v \neq \epsilon$, there exists a refinement $\hat{u} \cdot \underline{\hat{v}} \cdot \hat{w}$ of $u \cdot \underline{v} \cdot w$ such that one of the following two conditions holds:
\smallskip
\begin{itemize} \itemsep2mm
 \item $f(\hat{u} \cdot \underline{\hat{v}}^i \cdot \hat{w}) \ = \ f(\hat{u} \cdot \underline{\hat{v}}^{i+1} \cdot \hat{w})$ for every $i \geq N$.
 \item $f(\hat{u} \cdot \underline{\hat{v}}^i \cdot \hat{w}) \ < \ f(\hat{u} \cdot \underline{\hat{v}}^{i+1} \cdot \hat{w})$ for every $i \geq N$.
\end{itemize}
\end{thm}

Before going into the details of the proof let us show how to use the lemma.

\begin{exa}\label{ex:unambiguous}
We show that $f_2$ from Example~\ref{ex:min-letter} is not definable by any WA over $\natsemiring$.
Indeed, suppose it is definable and fix $N$ from Theorem~\ref{theorem:pumping_refined}. Consider the word $w = a^{(N+1)^2}\underline{b^{N}}$ and notice that $f_2(w) = N$. By refining $w$ we get $\hat{u} \cdot \underline{\hat{v}} \cdot \hat{w} = a^{(N+1)^2} b^n \underline{b^m} b^l$ for some $n, m, l$ such that $1 \leq m \leq N$ and $n + m + l = N$.
Since $n + m\cdot N +l < n + m\cdot(N+1) +l < (N+1)^2$ it must be the case that $f_2(\hat{u} \cdot \underline{\hat{v}}^i \cdot \hat{w}) < f_2(\hat{u} \cdot \underline{\hat{v}}^{i+1} \cdot \hat{w})$ for all $i \ge N$.
However, $f_2(\hat{u} \cdot \underline{\hat{v}}^{i} \cdot \hat{w}) = (N+1)^2$ for $i$ sufficiently large, which is a contradiction.
\end{exa}


\begin{exa}\label{ex:unambiguous_ok}
On the other hand, the function $f_1$ from Example~\ref{ex:in_unambiguous} satisfies Theorem~\ref{theorem:pumping_refined}. Consider a word $u\cdot \underline{v} \cdot w \in \Sigma^*$ and its refinement $\hat{u} \cdot \underline{\hat{v}} \cdot \hat{w}$. If $\hat{w}$ or $\hat{v}$ contain $b$ then $f(\hat{u} \cdot \underline{\hat{v}}^i \cdot \hat{w}) \ = \ f(\hat{u} \cdot \underline{\hat{v}}^{i+1} \cdot \hat{w})$ because the suffix of $a$'s remains the same. Otherwise, $f(\hat{u} \cdot \underline{\hat{v}}^i \cdot \hat{w}) \ < \ f(\hat{u} \cdot \underline{\hat{v}}^{i+1} \cdot \hat{w})$ since the suffix of $a$'s increases when pumping. Moreover, it is straightforward to generalise this argument and prove Theorem~\ref{theorem:pumping_refined} for all U-WA over $\trop$.
\end{exa}

To prove Theorem~\ref{theorem:pumping_refined} we use the following definitions.
For a matrix $M \in \natsemiring^{Q \times Q}$ recall that $\bar{M}$ is its homomorphic image in $\bool_\infty^{Q \times Q}$ (see Section~\ref{subsec:monoids}). We write that $M$ and $N$ in $\natsemiring^{Q \times Q}$ are equivalent, denoted $M \equiv_{\bool_\infty} N$, iff $\bar{M} = \bar{N}$.
We also extend the homomorphic image and equivalence relation from matrices to vectors.
We say that $D \in \natsemiring^{Q\times Q}$ is an \emph{idempotent} if $\bar{D}$ is an idempotent in the finite monoid $\bool_\infty^{Q \times Q}$.

\begin{lem}\label{lemma:equiv-pos}
	If $M \equiv_{\bool_\infty} N$, then $x^T \cdot M \cdot y > 0$ if and only if $x^T \cdot N \cdot y > 0$ for every $x, y \in \natsemiring^Q$.
\end{lem}
\begin{proof}
	Suppose that $x^T \cdot M \cdot y > 0$. By definition $x^T \cdot M \cdot y = \sum_{p,q} x(p) \cdot M(p,q) \cdot y(q)$.
	Then there exist $p, q \in Q$ such that $x(p) \cdot M(p,q) \cdot y(q) > 0$ and, in particular, $M(p,q) > 0$.
	Given that $M \equiv_{\bool_\infty} N$ we conclude $N(p,q) > 0$ and $x(p) \cdot N(p,q) \cdot y(q) > 0$, which proves $x^T \cdot N \cdot y > 0$.
\end{proof}

\begin{proof}[Proof of Theorem~\ref{theorem:pumping_refined}]
	Let $\cA = (Q, \AL, \{M_a\}_{a \in \Sigma}, I, F)$ be a WA over $\natsemiring$ such that $f = \sem{\cA}$.
	Without loss of generality, we assume that $I(q) \neq \infty$ and $M_a(p,q) \neq \infty$ for every $p,q \in Q$ and $a \in \Sigma$, namely, $\infty$ can only appear in the final vector $F$.
	Indeed, if $\infty$ is used in $I$ or some $M_a$, we can construct two weighted automata $\cA',\cA^\infty$ such that $\cA'$ is the same as $\cA$ but each $\infty$-initial state or each $\infty$-transition is replaced with $0$, and $\cA^\infty$ outputs $\infty$ if there exists some run in $\cA$ that outputs $\infty$ and 0 otherwise.
	Note that $\cA'$ has no $\infty$-transition or $\infty$-initial state and $\cA^\infty$ can be constructed in such a way that only the final vector contains $\infty$-values. The disjoint union of $\cA'$ and $\cA^\infty$ is equivalent to $\cA$.

	Let $N = \max\{|Q|, K\}$  where $K$ is the constant from Lemma~\ref{lemma:ramsey} for the finite monoid $\bool_\infty^{Q \times Q}$.
	For every word $u \cdot v \cdot w \in \Sigma^*$ such that $v = a_1 \ldots a_n$ with $n \geq N$, consider the output $I^T \cdot M_u \cdot M_v \cdot M_w \cdot F$ of $\cA$ over $u \cdot v \cdot w$.
	By Lemma~\ref{lemma:ramsey}, there exists a factorisation of the form:
	\[
	M_v = (M_{a_1} \cdot \ldots \cdot M_{a_i}) \cdot (M_{a_{i+1}} \cdot \cdots \cdot M_{a_{j}}) \cdot (M_{a_{j+1}} \cdot \ldots \cdot M_{a_{n}})
	\]
	for some $i < j$ where $M_{a_{i+1} \dots a_j}$ is an
        idempotent (i.e., $\bar{M}_{a_{i+1}\dots a_j}$ is an idempotent).
	We define the refinement $\hat{u} \cdot \underline{\hat{v}} \cdot \hat{w}$ of $u \cdot \underline{v} \cdot w$ such that $\hat{u} = u \cdot (a_1 \ldots a_i)$, $\hat{v} = a_{i+1} \ldots a_{j}$, and $\hat{w} = (a_{j+1} \ldots a_n) \cdot w$.
	Furthermore, define $x^T = I^T \cdot M_{u a_1 \dots a_i}$, $D =
        M_{a_{i+1} \dots a_{j}}$, and $y = M_{a_{j+1}\dots a_{n}w}\cdot F$.
	Note that $f(\hat{u} \cdot \hat{v}^i \cdot \hat{w}) = x^T \cdot D^i \cdot y$ for every $i \geq 0$ and $D$ is an idempotent (i.e., $\bar{D}$ is an idempotent).
	It remains to show the following lemma.

	\begin{lem}\label{lemma:nat-idempotent}
		For every idempotent $D \in \natsemiring^{Q \times Q}$ and $x, y\in \natsemiring^{Q}$ where $D$ and $x$ do not contain $\infty$-values, one of the conditions holds:
		\begin{align}
		& x^T \cdot D^i \cdot y \; = \; x^T \cdot D^{i+1} \cdot y  \;\;\; \text{ for every } i \geq |Q|, \;\;\; 	\text{ or } \label{eq:plus-lemma-cond1} \\
		& x^T \cdot D^i \cdot y \; < \; x^T \cdot D^{i+1} \cdot y  \;\;\; \text{ for every } i \geq |Q|. \label{eq:plus-lemma-cond2}
		\end{align}
	\end{lem}
We start showing that Lemma~\ref{lemma:nat-idempotent} holds when $y = e_p$ for some $p \in Q$, where $e_p(q) = 1$ if $q = p$ and $0$ otherwise. Note that $z = \sum_{p \in Q} z(p)\cdot e_p$ for every vector $z$.

	We say that $p$ is \emph{$D$-stable} (or just \emph{stable}) if $D(p,p) > 0$.
	Note that if $p$ is stable, then $D^i(p,p) > 0$ for every $i > 0$ (recall that $D$ is idempotent). Furthermore, $D \cdot e_p = e_p + z$ for some $z \in \natsemiring^{Q}$.
	Suppose that $p$ is stable and $D \cdot e_p = e_p + z$ for some vector $z$. Then for $i > 0$:
	$$
	\begin{array}{rclrcl}
		x^T \cdot D^{i+1} \cdot e_p & = & x^T \cdot D^{i} \cdot (e_p + z) & = & x^T \cdot D^{i} \cdot e_p + x^T \cdot D^{i} \cdot z
	\end{array}
	$$
	Given that $D$ is idempotent and $D^i \equiv_{\bool_\infty} D$, by Lemma~\ref{lemma:equiv-pos} we have that $x^T \cdot D^{i} \cdot z > 0$ if, and only if, $x^T \cdot D \cdot z > 0$.
	Therefore, if $x^T \cdot D \cdot z > 0$, we get that $x^T \cdot D^{i} \cdot e_p < x^T \cdot D^{i+1} \cdot e_p$ for every $i > 0$, in particular, for every $i \geq |Q|$. Otherwise, $x^T \cdot D \cdot z = 0$ and $x^T \cdot D^{i} \cdot e_p = x^T \cdot D^{i+1} \cdot e_p$ for every $i > 0$, in particular, for every $i \geq |Q|$.

	Consider the relation $\preceq_{D} \mathop{\subseteq} Q\times Q$ such that $p \preceq_{D} q$ if $p = q$ or $D(p,q) > 0$. Let $P \subseteq Q$ be the set of all non-stable states in $D$.
	One can easily check that $\preceq_{D}$ restricted to $P$ forms a partial order, namely, that $\preceq_{D}$  is reflexive, antisymmetric, and transitive. Indeed, transitivity holds because $D$ is idempotent.
	To prove antisymmetry, note that for every non-stable states $p$ and $q$, if $p \preceq_{D} q$, $q \preceq_{D} p$ and $p \neq q$ hold, then $D(p,p) > 0$.
	This is a contradiction since $p$ is non-stable.

	Since $\preceq_D$ is a partial order, we prove the lemma for $y = e_p$ by induction over $\preceq_D$.
	Formally, we strengthen the inductive hypothesis such that conditions~\eqref{eq:plus-lemma-cond1} and~\eqref{eq:plus-lemma-cond2} hold for every $i \geq N_{q}$, where $N_q = |\{q' \in P \ | \  q' \preceq_D q\}|$ (notice that $N_q \le |Q|$ for every $q$).
	The base case is for $N_p = 1$, which means that $p$ is stable. For the inductive case suppose that $N_p > 1$. Then
	$$
	x^T \cdot D^{i+1} \cdot e_p \ = \ x^T \cdot D^{i} \cdot (c_1 \cdot e_{q_1} + \ldots + c_k \cdot e_{q_k}) \ = \ c_1(x^T \cdot D^{i} \cdot e_{q_1}) + \ldots + c_k(x^T \cdot D^{i} \cdot e_{q_k})
	$$
	for pairwise different states $q_1, \ldots, q_k$ and positive values $c_1, \ldots, c_k \in \nat$ such that $q_j$ is either stable or $q_j \prec_D p$.
	Thus all states $q_1, \ldots, q_k$ satisfy our inductive hypothesis.

	Consider the partition of $q_1, \ldots, q_k$ into sets $C_{=}$ and $C_{<}$ such that $C_{=}$ and $C_<$ satisfy condition \eqref{eq:plus-lemma-cond1} and \eqref{eq:plus-lemma-cond2}, respectively.
	If $C_{<} = \emptyset$, then for every $i \geq N_p$ we have:
\begin{align}
	x^T \cdot D^{i+1} \cdot e_p \;\;& = \;\; c_1(x^T \cdot D^{i} \cdot e_{q_1}) + \ldots + c_k(x^T \cdot D^{i} \cdot e_{q_k}) \nonumber \\
& = \;\; c_1(x^T \cdot D^{i-1} \cdot e_{q_1}) + \ldots + c_k(x^T \cdot D^{i-1} \cdot e_{q_k}) \nonumber  \\
& = \;\; x^T \cdot D^{i} \cdot e_p \label{eq:last}.
\end{align}
	Note that $x^T \cdot D^{i} \cdot e_{q_j} = x^T \cdot D^{i-1} \cdot e_{q_j}$ holds by the inductive hypothesis and because $N_p > N_{q_j}$ for every $q_j$.
	Suppose otherwise, that $C_{<} \neq \emptyset$ and there exists a state $q_j$ that satisfies $x^T \cdot D^i \cdot e_{q_j} < x^T \cdot D^{i+1} \cdot e_{q_j}$ for every $i \geq N_{q_j}$. Then it is straightforward that equality~\eqref{eq:last} becomes a strict inequality and condition~\eqref{eq:plus-lemma-cond2} holds.

	We have shown that either~(\ref{eq:plus-lemma-cond1}) or~(\ref{eq:plus-lemma-cond2}) holds for $y = e_p$. It remains to extend this to any vector~$y \in \natsemiring^Q$ (possibly with $\infty$). Note that
	$$
	x^T \cdot D^{i+1} \cdot y \ = \ y(q_1) \cdot (x^T \cdot D^{i+1} \cdot e_{q_1}) + \ldots + y(q_k)\cdot (x^T \cdot D^{i+1} \cdot e_{q_k})
	$$
	for some states $q_1, \ldots, q_k$ such that $y(q_j) > 0$ for every $j \leq k$.
	We consider two cases.
	First, if there exists $j$ such that $y(q_j) = \infty$ and $x^T \cdot D^{i} \cdot e_{q_j} > 0$ for $i \geq N$, then $x^T \cdot D^{i} \cdot y = \infty$ for every $i \geq 0$.
	Thus, $x^T \cdot D^{i} \cdot y$ satisfies condition~(\ref{eq:plus-lemma-cond1}).
	Second, suppose that for every $j$ we have  $y(q_j) \neq \infty$ or $x^T \cdot D^{i} \cdot e_{q_j} = 0$ for $i \geq N$. It suffices to consider the case when $y(q_j) \neq \infty$ for all $j$. Then if some $x^T \cdot D^{i} \cdot e_{q_j}$ satisfies condition~(\ref{eq:plus-lemma-cond2}) we have that $x^T \cdot D^{i} \cdot y$ satisfies condition~(\ref{eq:plus-lemma-cond2}).
	Conversely, if every $x^T \cdot D^{i} \cdot e_{q_j}$ satisfies condition~(\ref{eq:plus-lemma-cond1}) we have that $x^T \cdot D^{i} \cdot y$ satisfies condition~(\ref{eq:plus-lemma-cond1}).
\end{proof}

One could try to simplify Theorem~\ref{theorem:pumping_refined} changing the condition $i \geq N$ to $i \ge 0$. Unfortunately, we do not know if the theorem would remain true. A naive approach would be to use a generalisation of Lemma~\ref{lemma:ramsey}, but intuitively, the behaviour of non-stable states is problematic.
We conclude with the following remarks, straightforward from the proof. We will use them in Section~\ref{sec:fin_min}.
	\begin{rem}
\label{remark:pluslemma2}
Changing $y$ to $y'$ such that $y \equiv_{\bool_\infty} y'$ does not influence whether condition~\eqref{eq:plus-lemma-cond1} or condition~\eqref{eq:plus-lemma-cond2} holds in Lemma~\ref{lemma:nat-idempotent} (notice that here we need that the abstractions have values in $\bool_\infty$ not in $\bool$). Similarly, changing $x$ to $x'$ such that $x \equiv_{\bool_\infty} x'$ does not influence whether condition~\eqref{eq:plus-lemma-cond1} or~\eqref{eq:plus-lemma-cond2} holds.
\end{rem}

\begin{rem}
\label{remark:pluslemma}
The constant $N$ and the refinement of $w$ depend only on the finite
monoid $\bool_\infty^{Q \times Q}$. In particular they are independent
from the initial and final vectors $I$ and $F$.
\end{rem}

\section{Finite-min recognisable functions}
\label{sec:fin_min}
In this section we focus on recognisable functions over $\natsemiring$ with
some $\min$ operations allowed.
Formally, we say that $f: \Sigma^* \rightarrow \natinf$ is a finite-min recognisable function, if there exist recognisable functions $f_1, \ldots, f_m$ over $\natsemiring$ such that $f(w) = \min\{f_1(w), \ldots, f_m(w)\}$.
It is known that over $\trop$, FA-WA are equivalent to a finite minimum of
U-WA~\cite{Weber94}, hence the
functions defined by FA-WA  are included in the class of finite-min recognisable functions.
As a corollary of the pumping lemma in this section we show that PA-WA are strictly more expressive than FA-WA over $\trop$ (Example~\ref{ex:finite-min} and Example~\ref{ex:finite-min2}).

We start by introducing some notation. Generalising the notation used
in the previous section, we define for $n >0$ an \emph{$n$-pumping
  representation}  for a
word $w \in \Sigma^*$ as a factorisation of the form
\[
w = u_0 \cdot \underline{v_1} \cdot u_1 \cdot \underline{v_2} \cdot \ldots u_{n-1} \cdot \underline{v_n} \cdot u_n,
\]
where $w = u_0 \cdot v_1 \cdot u_1 \cdot v_2 \cdot \ldots v_n \cdot u_n$ and $v_k \neq \epsilon$ for all $k$.
A refinement of an $n$-pumping representation for $w$ is given by
\[
w = u_0' \cdot \underline{y_1} \cdot u_1' \cdot \underline{y_2} \cdot \ldots u_{n-1}' \cdot \underline{y_n} \cdot u_n',
\]
if $v_k = x_k \cdot y_k \cdot z_k$, $u_k' = z_k \cdot u_k \cdot x_{k+1}$; where $z_0 = x_{n+1} = \epsilon$ and $y_k \neq \epsilon$ for every $k$.
Let $S \subseteq \{1, \ldots, n\}$ such that $S \neq\emptyset$. Let
$\underline{y_k}$ be a factor of the refined $n$-pumping
representation of $w$. By $\underline{y_k}(S,i)$ we denote the word $y_k^i$ if $k \in S$ and $y_k$ otherwise.
By $w(S,i)$ we denote the word
$$
w = u'_0 \cdot \underline{y_1}(S,i) \cdot u'_1 \cdot \underline{y_2}(S,i) \cdot \ldots u'_{n-1} \cdot \underline{y_n}(S,i) \cdot u'_n.
$$
In other words, in  $w(S,i)$ we pump $i$ times each factor $y_k$, for
all $k \in S$. Note that the pumping always refers to the
\emph{refinement} of the $n$-pumping representation.
%

\begin{thm}[Pumping Lemma for finite-min recognisable functions]\label{theorem:some_min}
Let $f : \Sigma^* \to \natinf$ be a finite-min recognisable function. There exists $N$ such that for every $n$-pumping representation
\[
w = u_0 \cdot \underline{v_1} \cdot u_1 \cdot \underline{v_2} \cdot \ldots u_{n-1} \cdot \underline{v_n} \cdot u_n,
\]
where $|v_i| \ge N$ for all $i$, there exists a refinement
\[
w = u_0' \cdot \underline{y_1} \cdot u_1' \cdot \underline{y_2} \cdot \ldots u_{n-1}' \cdot \underline{y_n} \cdot u_n',
\]
such that for every sequence of non-empty, pairwise different subsets $S_1, \dots, S_{k} \subseteq \{1, \ldots, n\}$ with $k \geq N$ one of the following holds:
\begin{itemize}
 \item there exists $j$ such that $f(w(S_j, i)) < f(w(S_j, i+1))$ for
   all but finitely many $i$;
 \item there exist $j_1 \neq j_2$ such that $f(w(S_{j_1} \cup S_{j_2},
   i)) = f(w(S_{j_1} \cup S_{j_2}, i+1))$ for all but finitely many $i$.
\end{itemize}
\end{thm}

Before proving Theorem~\ref{theorem:some_min}, we show how to use it with two examples.

\begin{exa}\label{ex:finite-min}
We show that $f_3$ from Example~\ref{ex:min-a_to_b} is not definable by finite-min recognisable functions. Indeed, fix $N$ from Theorem~\ref{theorem:some_min} and consider the $n$-pumping representation $w = (\underline{b^N} \cdot \underline{a^N})^N$.
We index each pumping factor with a pair $(s, j)$, where $j \le N$
denotes the block $b^N a^N$, and $s \in \{1,2\}$ denotes the factor in the block.
First, notice that $f_3(w) = N \cdot (N-1)$ because runs minimising the value for $\mathcal{W}_3$ change the state after reading the last $b$ in one of the blocks. We define the sets $S_j = \{(1,j), (2,j)\}$ for $j \in \{1,\ldots, N\}$. Clearly $f_3(w(S_j,i)) = f_3(w)$ for all $j$ and $i$, because the run minimising the value changes the state after the last $b$ in the $j$-th block. On the other hand $f_3(w(S_{j_1} \cup S_{j_2}, i)) < f_3(w(S_{j_1} \cup S_{j_2}, i+1))$ for all $i$ and $j_1 \neq j_2$. Hence $f_3$ does not satisfy the pumping lemma for finite-min recognisable functions.
\end{exa}

\begin{exa}\label{ex:finite-min2}
We show that $f_4$ from Example~\ref{ex:min-b-substrings} is not definable by finite-min recognisable functions. Indeed, fix $N$ from Theorem~\ref{theorem:some_min}. Consider the $N$-pumping representation $w = (\underline{b^N}a)^N$. Then by definition $f_4(w) = N$. In the refinement all pumping parts will be of the form $b^{n}$ for $1 \le n \le N$. We define the sets $S_j = \{1, \ldots, N\} \setminus \{j\}$ for all $1 \le i \le N$. Clearly $f_4(w(S_j, i)) = N$ for all $j$ and $i$. On the other hand $f_4(w(S_{j_1} \cup S_{j_2}, i)) < f_4(w(S_{j_1} \cup S_{j_2}, i+1))$ for all $i$ and $j_1 \neq j_2$. Hence $f_4$ does not satisfy the pumping lemma for finite-min recognisable functions.
\end{exa}

\begin{proof}[Proof of Theorem~\ref{theorem:some_min}]
Let $f_1, \ldots, f_m$ be recognisable functions over $\natsemiring$ such that $f(w) = \min\{f_1(w), \ldots, f_m(w)\}$ for every $w$.
Furthermore, consider $\cA_j = (Q_j, \Sigma, \{M_{j,a}\}_{a \in \Sigma}, I_j, F_j)$ the corresponding WA for $f_j$.
Let $Q = \bigcup_j Q_j$ (we assume that $Q_1, \ldots, Q_m$ are pairwise disjoint) and consider the set of matrices $\{U_a\}_{a \in \Sigma}$ where $U_a \in \natsemiring^{Q \times Q}$ such that
$U_a(p,q) = M_{j,a}(p,q)$ whenever $p, q \in Q_j$ and $0$ otherwise.
Then $f_j(w) = (I_j')^t \cdot U_w \cdot F_j'$ for every $j$ and $w \in \Sigma^*$ where $I_j'$ and $F_j'$ are the extensions of $I_j$ and $F_j$ from $Q_j$ into $Q$ such that $I_j'(q) = I_j(q)$ and $F_j'(q) = F_j(q)$ whenever $q \in Q_j$ and $0$ otherwise.
Notice that $\{U_a\}_{a \in \Sigma}$ synchronise the behaviour of $f_1, \ldots, f_m$ in a single set of matrices and project the output of $f_j$ with $I_j'$ and $F_j'$.
Let $N = \max\{K, m+1\}$ such that $K$ is the constant from Lemma~\ref{lemma:ramsey} applied to $\bool_\infty^{Q \times Q}$.

Let
$
w = u_0 \cdot \underline{v_1} \cdot u_1 \cdot \underline{v_2} \cdot \ldots u_{n-1} \cdot \underline{v_n} \cdot u_n.
$
be an $n$-pumping representation as in the statement of the theorem.
For every $i$ we apply Theorem~\ref{theorem:pumping_refined} to
$u_{\leq i} \cdot \underline{v_i} \cdot t_{\geq i}$, where $u_{\leq i}
= u_0 \cdot v_1 \cdot \ldots u_{i-1}$ and $t_{\geq i} = u_i \cdot
v_{i+1} \cdot \ldots u_n$, and $\{U_a\}_{a \in \Sigma}$ (recall that
the refinement of $u_{\leq i} \cdot \underline{v_i} \cdot t_{\geq i}$
depends only on $\{U_a\}_{a \in \Sigma}$, and not on the initial or
final vector, see Remark~\ref{remark:pluslemma}). As in the proof of
Theorem~\ref{theorem:pumping_refined} we obtain a refinement
\[
w \ = \  u_0' \cdot \underline{y_1} \cdot u_1' \cdot \underline{y_2} \cdot \ldots u_{n-1}' \cdot \underline{y_n} \cdot u_n',
\]
where each $y_i$ is idempotent w.r.t.~$\{U_a\}_{a \in \Sigma}$.

Note that the refinement is the same for each function $f_j$.
Therefore, we obtain
\[
f_j(w) \ = \  (I_j')^t \cdot U_{u_0'} \cdot D_1 \cdot \ldots \cdot U_{u_{n-1}'} \cdot D_n \cdot  U_{u_n'} \cdot F_j'
\]
where all $D_i = U_{y_i}$ are idempotents.

\begin{lem}
\label{lemma:plus_only}
Let $S \subseteq \{1,\ldots,n\}$ be a non-empty set and fix one function $f_j$. Then $f_j(w(S,i)) < f_j(w(S,i+1))$ for every $i \ge N$ iff there exists $k \in S$ such that $f_j(w(\{k\},i)) < f_j(w(\{k\},i+1))$ for every $i \ge N$.
\end{lem}

\begin{proof}
By definition
$
f_j(w(S,i)) =  (I_j')^t \cdot U_{u_0'} \cdot D_1^{s_1} \cdot \ldots \cdot U_{u_{n-1}'} \cdot D_n^{s_n} \cdot  U_{u_n'} \cdot F_j'
$
where $s_k = i$ if $k \in S$ and $s_k = 1$ otherwise. Since all $D_i$ are idempotents then for all $k$: 
\begin{align*}
(I_j')^t \cdot U_{u_0'} \cdot D_1^{s_1} \cdot \ldots \cdot D_{k-1}^{s_{k-1}} \cdot U_{u_{k-1}'} & \ \equiv_{\bool_\infty} \ (I_j')^t \cdot U_{u_0'} \cdot D_1 \cdot \ldots \cdot D_{k-1} \cdot U_{u_{k-1}'} \\
U_{u_{k}'} \cdot D_{k+1}^{s_{k+1}} \cdot \ldots \cdot D_n^{s_n} \cdot  U_{u_n'} \cdot F_j' & \  \equiv_{\bool_\infty} \  U_{u_{k}'} \cdot D_{k+1} \cdot \ldots \cdot D_n \cdot  U_{u_n'} \cdot F_j'.
\end{align*}
Hence, the lemma follows from Remark~\ref{remark:pluslemma2}.
\end{proof}

To finish the proof we analyse $f(w(S, i)) = \min\{f_1(w(S, i)), \ldots , f_m(w(S, i))\}$. Consider a sequence of subsets $S_1, \ldots, S_{k}$ with $k \geq N$.
Suppose there is a set $S_l$ 
such that for every $1 \leq j \leq
m$, 
we have $f_j(w(S_l,i)) < f_j(w(S_l,i+1))$ for every $i \ge N$. In this
case,
$f(w(S_l,i)) < f(w(S_l,i+1))$ holds for
all $i \ge N$, so the first condition of the theorem
is met.

Suppose otherwise that no such $S_l$ exists. In particular, for every
$S_l$ there is at least one $j$ such that
$f_j(w(\{s\},i))=f_j(w(\{s\},i+1))$ for all $i \ge N$ and all $s \in S_l$, hence
$f_j(w(S_l,i))=f_j(w(S_l,i+1))$ for all $i \ge N$. For every $S_l$ let
$X_l \subseteq \{1, \ldots, m\}$ be the set of indices $j$ such that
$f_{j}(w(S_l,i)) = f_{j}(w(S_l,i+1))$ for all $i \ge N$. By the above
assumptions, every $X_l$ is non-empty. Since $k \geq N > m$ there
exists $l_1,l_2$ such that $X_{l_1} \cap X_{l_2} \neq \emptyset$. From
Lemma~\ref{lemma:plus_only} it follows that for $i \ge N$ it holds: $f_j(w(S_{l_1} \cup S_{l_2},i)) = f_j(w(S_{l_1} \cup S_{l_2},i+1))$ for all $j \in X_{l_1} \cap X_{l_2}$; and $f_j(w(S_{l_1} \cup S_{l_2},i)) < f_j(w(S_{l_1} \cup S_{l_2},i+1))$ for all $j \in \{1,\ldots, m\} \setminus (X_{l_1} \cap X_{l_2})$. Hence for $i$ sufficiently large $f(w(S_{l_1} \cup S_{l_2}, i)) = \min_{j \in X_{l_1} \cap X_{l_2}}(f_j(w(S_{l_1} \cup S_{l_2},i)))=f(w(S_{l_1} \cup S_{l_2}, i+1))$, which concludes the proof.
\end{proof}

\section{Poly-ambiguous recognisable functions over the min-plus semiring}
\label{sec:pumping}
In this section we focus on polynomial-ambiguously recognisable functions
over $\trop$. We expect that there is a wider class of functions,
definable like in the previous section, where
Theorem~\ref{theorem:polynomial} holds, but this is left for future
work. A consequence of this section is that WA are strictly more
expressive than PA-WA (see Examples~\ref{ex:polynomial} and~\ref{ex:polynomial2}).

We will use in the following the notation of $n$-pumping representations from Section~\ref{sec:fin_min}.
A sequence of non-empty sets $S_1, \ldots, S_m$ over $\{1, \ldots, n
\}$ is called a \emph{partition} if the sets are pairwise disjoint and
their union is $\{1,\ldots,n\}$.
Furthermore, we say that $S \subseteq \{1, \ldots, n\}$ is a \emph{selection set} for $S_1, \ldots, S_m$ if $|S \cap S_i| = 1$ for every $i$.

\begin{thm}[Pumping Lemma for polynomially-ambiguous automata]
	\label{theorem:polynomial}
	Let $f : \Sigma^* \to \natinf$ be a polynomial-ambiguously recognisable function over $\trop$. There exist $N$ and a function $\varphi: \nat \to \nat$ such that for every $n$-pumping representation:
	\[
	w = u_0 \cdot \underline{v_1} \cdot u_1 \cdot \underline{v_2} \cdot \ldots \cdot u_{n-1} \cdot \underline{v_n} \cdot u_n,
	\]
	where $|v_i| \geq N$ for every $i \leq n$, there exists a
        refinement:
	\[
	w = u_0' \cdot \underline{y_1} \cdot u_1' \cdot \underline{y_2} \cdot \ldots u_{n-1}' \cdot \underline{y_n} \cdot u_n',
	\]
	such that for every partition $\pi = S_1,\ldots,S_m$ of $\{1, \ldots, n\}$ with $m \geq \varphi(\max_j(|S_j|))$, one of the following holds:
	\begin{itemize}
		\item there exists $j$ such that $f(w(S_j, i)) =
                  f(w(S_j, i+1))$ for all but finitely many $i$;
		\item there exists a selection set $S \subseteq
                  \{1,\ldots,n\}$ for $\pi$ such that $f(w(S, i)) <
                  f(w(S, i+1))$ for all but finitely many $i$.
	\end{itemize}
\end{thm}

\begin{exa}\label{ex:polynomial}
We show that $f_5$ from Example~\ref{ex:no-poly-ambiguous} is not
definable by any PA-WA.
Indeed, let $N$ and $\varphi$ be the constant and the function from Theorem~\ref{theorem:polynomial}. Consider the following $2m$-pumping representation:
$w = (\underline{a^N} \cdot \underline{b^N}\#)^m$ where $m \geq
\varphi(2)$ (here $\max_i(|S_i|)$ will be equal to 2).
We index the $j$-th block of $a$'s with $j$ and the $j$-th block of $b$'s with $j'$.
We define the subsets $S_1, \ldots, S_m$ as  $S_j = \{j,j'\}$. Clearly, for all $j$ we have $f_5(w(S_j,i)) < f_5(w(S_j,i+1))$. On the other hand for every selection set $S$ we have $f_5(w(S,i)) = f_5(w(S, i+1))$.
Hence $f_5$ does not satisfy Theorem~\ref{theorem:polynomial}.
\end{exa}

\begin{exa}\label{ex:polynomial2}
The function $f_5$ in Example~\ref{ex:no-poly-ambiguous} is essentially the function $f_2$ from Example~\ref{ex:min-letter} applied to the subwords between the symbols $\#$, where the outputs are aggregated with $+$. In a similar way one can define a min-plus automaton recognising $f_6(w) = \sum_{i}f_4(w_i)$ for any $w \in \Sigma^*$ of the form $w_0 \# w_1 \# \ldots \# w_n$ with $w_i \in \{a,b\}^*$, where $f_4$ is the function computing the minimal block of $b$'s from Example~\ref{ex:min-b-substrings}. We show that $f_6$ is not definable by PA-WA over $\trop$. Consider the following $2m$-pumping representation:
$w = (\underline{b^N} \cdot a \cdot \underline{b^N}\#)^m$ where $m
\geq \varphi(2)$ (here $\max_i(|S_i|)$ is again 2).
As in Example~\ref{ex:polynomial}, we index the first $j$-th block of
$b$'s with $j$ and the second $j$-th block of $b$'s with $j'$, and
we set $S_j = \{j,j'\}$, for $1 \leq j \leq m$.
Clearly, for all $j$ we have $f_6(w(S_j,i)) < f_6(w(S_j,i+1))$. On the
other hand for every selection set $S$ we have $f_6(w(S,i)) = f_6(w(S,
i+1))$. Hence $f_6$ does not satisfy Theorem~\ref{theorem:polynomial}
either.
\end{exa}

Consider the set of matrices $\trop^{Q \times Q}$ over the min-plus semiring. Recall that here $\add = \min$, $\mult = +$, $\zero = \infty$, $\one = 0$, and the product of matrices $M, N \in \trop^{Q \times Q}$ is defined by $M \cdot N(p,q) = \min_{r}(M(p,r) + N(r,q))$.
Also, recall that for any $M \in \trop^{Q \times Q}$ we denote by $\bar{M}$ the homomorphic image of $M$ into the finite monoid $\bool^{Q \times Q}$ (see Section~\ref{subsec:monoids}).
Similar as in Section~\ref{sec:pump-sum} and Section~\ref{sec:fin_min}, we say that $D \in \trop^{Q \times Q}$ is an idempotent if $\bar{D}$ is an idempotent in the finite monoid $\bool^{Q\times Q}$.

The following lemma states a special property of polynomially-ambiguous automata that we exploit in the proof of Theorem~\ref{theorem:polynomial}.

\begin{lem}
	\label{lemma:polymatrix}
	Let  $\cA = (Q, \Sigma, \{M_a\}_{a \in \Sigma}, I, F)$ be a polynomially-ambiguous weighted automaton over the min-plus semiring.
	For every idempotent $D \in \{M_w \mid w \in \Sigma^*\}$ and
        for every $p, q \in Q$, there exist constants $c, d \in \trop$ and $b \in \nat$ such that
	$
	D^{b+i}(p, q) = c\cdot i + d
	$ for all $i \ge 0$.
\end{lem}

\begin{proof}
We can view  $\bar{D} \in \bool^{Q \times Q}$ as the adjacency matrix of a graph.
Now we show that the cycles of the directed graph defined by
$\bar{D}$ can be only self-loops.
Indeed, assume by contradiction that there exists a cycle passing through $r,s \in Q$ with
$r \neq s$ then $\bar{D}(r, s) = \bar{D}(s, r) = \bar{D}(r,r) =
\bar{D}(s, s) = 1$ (because $\bar{D}$ is an idempotent). Since $D \in
\{M_w \mid w \in \Sigma^*\}$, it can be verified that $\cA$ cannot be
polynomially-ambiguous~\cite{weber1991degree}. Indeed, let $w$ be the word such that $D = M_w$.
Then there are at least two different paths from $r$ to $r$ when reading $w^2$. Hence when reading $w^{2n}$ the
number of paths is at least $2^n$. Since all automata considered in this paper are trimmed the state $r$ is both
accessible and co-accessible. Hence for every $n$ there is a word of length linear in $n$ with at least $2^n$ accepting runs.
This is a contradiction with the assumption that $\cA$ is polynomially-ambiguous.

We conclude that $\bar{D}$ forms an acyclic graph
with some self-loops and the states in $Q$ can be ordered as $p_1,
\ldots, p_n$, such that $\bar{D}(p_j, p_i) = +\infty$ for every $i < j$.

If $D(p, q) = + \infty$ then it suffices to take $c = d = +\infty$ since $D$ is an idempotent.
Otherwise, $D^i(p, q) = \min\left( \sum_{1 \le k \le i}D(p_{j_{k-1}},
  p_{j_k}) \right)$, where the minimum is over all sequences $(j_k)_{k=0}^i$
such that $p_{j_0} = p$, $p_{j_i} = q$. Since $D(p, q) < + \infty$ we can restrict to sequences such that $D(p_{j_{k-1}},p_{j_k}) < + \infty$ for all $k$. Let $A_i$ be the set of indices such that $k \in A_i$ if: $0 < k \le i$; $D(p_{j_{k-1}}, p_{j_k}) < + \infty$; and $p_{j_{k-1}} \neq p_{j_k}$. Since $\bar{D}$ is acyclic it follows that $|A_i| \le |Q|$. In particular this means that the number of possible sets $A_i$ depends only on $|Q|$, not on $i$.

Suppose a sequence $(j_k)_{k=0}^i$ is a witness for the value of $D^i(p, q)$. Let $k \not \in A_i$ such that $0 < k \le i$. Then $D(p_{j_k},p_{j_k}) \le D(p_{j_s},p_{j_s})$ for all $0 < s \le i$. Thus we can modify $(j_k)_{k=0}^i$ into a witness such that for all $k \not \in A_i$ such that $k > 0$ the states $p_{j_k}$ are all the same. We denote this state $p_j$.
Notice that for $i$ big enough the value of $D^i(p, q)$ depends mostly on $D(p_{j}, p_j)$. Then $p_j$ can be only one of the states $r$ such that $c = D(r, r)$ is minimal. In particular $c$ does not depend on $i$. Let $d_i = \sum_{k \in A_i} D(p_{j_{k-1}}, p_{j_k})$. Then $D^i(p, q) = d_i + c \cdot (i - |A_i|) = d_i - c\cdot |A_i| + c \cdot i$. Therefore for $i$ big enough, denoted $i \ge i_0$, this sum is achieved for $A_i$ such that $d_i - c \cdot |A_i|$ is minimal. Since the number of different $A_i$ is bounded we can assume that $A_i$ and $d_i$ do not depend on $i$, for $i \ge i_0$, and denote them $A$ and $d'$, respectively.

The lemma follows by fixing $c$ as above, $b = |A| + i_0$, and $d = d' + c \cdot i_0$.
\end{proof}

\begin{proof}[Proof of Theorem~\ref{theorem:polynomial}]
Consider a polynomially-ambiguous WA $\cA = (Q, \Sigma, \{M_a\}_{a \in \Sigma}, I, F)$ over $\trop$ such that $f = \sem{\cA}$.
We take for $N$ the constant from Lemma~\ref{lemma:ramsey} for the finite monoid $\bool^{Q \times Q}$.
The function $\varphi: \nat \rightarrow \nat$ will be determined later in the proof.

Consider an $n$-pumping representation $w$ like in the statement of the theorem.
By Lemma~\ref{lemma:ramsey}, for every $v_k$ there exists a factorisation $v_k = x_k y_k z_k$ such that $M_{y_k}$ is an idempotent and $|y_k| \le N$. We denote $D_k = M_{y_k}$ and define:
\[
w \; = \;  u_0' \cdot \underline{y_1} \cdot u_1' \cdot \underline{y_2} \cdot \ldots u_{n-1}' \cdot \underline{y_n} \cdot u_n'
\]
such that each word $y_k$ is the factor of $v_k$ corresponding to the idempotent $D_k$. In the remainder of the proof we denote $w_{\le k} = u_0' \cdot y_1 \cdot \ldots u_{k-1}'$. For every $S \subseteq \{1 \ldots n\}$ we denote by $w_{\le k}(S,i)$ the word $w_{\le k}$ with all $y_j$ pumped $i$ times for all $j < k$ such that $j \in S$.

Recall that $\Run_\cA(w)$ is the set of all accepting runs on $w$, and let $\rho \in \Run_\cA(w)$. Every run induces two states for each $1 \le k \le n$: states preceding and following each word $y_k$. In the rest of the proof these will be the most important parts of a run.
To work with them, we define the abstraction of $\rho$, denoted by
$\bar{\rho}: \{1, \ldots, n\} \to Q \times Q$, such that $\bar{\rho}(k)
= (p, q)$ where $p$ and $q$ are the states of $\rho$ reached after
$w_{\leq k}$ and $w_{\leq k}\cdot y_k$, respectively.
Similarly, for $S \subseteq \{1, \ldots, n\}$, $i \geq 1$, and $\rho \in
\Run_\cA(w(S,i))$ we define $\bar{\rho}: \{1, \ldots, n\} \to Q \times
Q$ such that $\bar{\rho}(k) = (p, q)$ where $p$ and $q$ are the states
of $\rho$ reached after $w_{\le k}(S,i)$ and $w_{\le k}(S,i) \cdot
y_k(S, i)$, respectively.
We denote by $\overline{\Run_\cA}(w)$ the set of all $\bar{\rho}$ with
$\rho \in \Run_\cA(w)$, and same for $\overline{\Run_\cA}(w(S,i))$.
Observe that since all $D_k$ are idempotents, $\overline{\Run_\cA}(w(S,i)) = \overline{\Run_\cA}(w)$ for all subsets $S$ and $i \ge 1$.

The next step is to prove that the cardinality of
$\overline{\Run_\cA}(w)$ is bounded by a polynomial
$P(\cdot)$ depending only on $\cA$, namely such that
$|\overline{\Run_\cA}(w)| \leq P(n)$.
Let $w'$ be the word obtained from $w$ where each $u_i'$ is replaced with a word $u_i''$ of length at most $|\bool^{Q \times Q}|$ such that $\overline{M_{u_i'}} = \overline{M_{u_i''}}$. To see that $u_i''$ exists suppose that $u_i' = a_1 \ldots a_s$ and $s > |\bool^{Q \times Q}|$. By the pigeonhole principle in the sequence of matrices $\overline{M_{a_1}}, \overline{M_{a_1a_2}}, \ldots, \overline{M_{a_1\ldots a_s}}$ there is a repetition that can be removed giving a shorter $u_i''$ such that$\overline{M_{u_i'}} = \overline{M_{u_i''}}$.
Then $|\Run_\cA(w')| \ge |\overline{\Run_\cA}(w)|$.
Recall that $|y_i| \le N$ and that $N$ depends only on $\bool^{Q \times Q}$.
Then by definition $|w'| \le (N + |\bool^{Q \times Q}|)\cdot(n+1)$ and thus $|\Run_\cA(w')| \le R((N + |\bool^{Q \times Q}|)\cdot(n+1))$, where~$R$ is the polynomial bounding the number of runs in $\cA$. Recall that $\cA$ is polynomially-ambiguous. The claim follows for $P(n) = R((N + |\bool^{Q \times Q}|)\cdot(n+1))$.

Fix a non-empty set $S \subseteq \{1, \ldots, n\}$ and a run $\rho \in \Run_{\cA}(w)$.
For every $k \in S$ let $b_{\bar{\rho}}^k$, $c_{\bar{\rho}}^k$
and $d_{\bar{\rho}}^k$ be the constants from
Lemma~\ref{lemma:polymatrix} such that $D_k^{b_{\bar{\rho}}^k +
  i}[\bar{\rho}(k)] = c_{\bar{\rho}}^k \cdot i +
d_{\bar{\rho}}^k$ for $i \ge 0$. Since $\rho$ is accepting, $c_{\bar{\rho}}^k, d_{\bar{\rho}}^k < + \infty$.
We show that:
\begin{enumerate}
 \item\label{eq:equal} $\sem{\cA}(w(S, i)) = \sem{\cA}(w(S, i+1))$ for
   all sufficiently large $i$ iff there exists a run $\rho \in \Run_{\cA}(w)$ such that $c_{\bar{\rho}}^k = 0$ for every $k \in S$;
 \item\label{eq:nequal} $\sem{\cA}(w(S, i)) < \sem{\cA}(w(S, i+1))$
   for all sufficiently large $i$ iff for every run $\rho \in \Run_{\cA}(w)$ there exists $k$ such that $c_{\bar{\rho}}^k > 0$.
\end{enumerate}

Since the number of different $b_{\bar{\rho}}^k$ is bounded we can assume that they are all equal to some $i_0$ by choosing the maximal $b_{\bar{\rho}}^k$.
Let $\rho \in \Run_\cA(w(S,i+1))$ be a run realising the minimum value for $i \ge i_0$.
Given that $D_k$ is idempotent one can obtain a run $\rho' \in
\Run_{\cA}(w(S,i))$ such that $\bar{\rho}' = \bar{\rho}$ by removing
one copy of each $y_k$. In particular
$|\rho'| \le |\rho|$, which proves $\sem{\cA}(w(S, i)) \le
\sem{\cA}(w(S, i+1))$. It follows that it suffices to
show~\eqref{eq:equal} above.

To prove~\eqref{eq:equal} suppose first that $\sem{\cA}(w(S, i)) =
\sem{\cA}(w(S, i+1))$ for all sufficiently large $i$. Let $\rho \in
\cA(w(S, i+1))$ and $\rho' \in \cA(w(S, i))$ be the previous runs
realising the minimum on $w(S,i+1)$ and its shortening, respectively.
By Lemma~\ref{lemma:polymatrix} $D_k^{i_0 + i+1}[\bar{\rho}(k)] =  c_{\bar{\rho}}^k \cdot (i+1) + d_{\bar{\rho}}^k$. If $c_{\bar{\rho}}^k > 0$ for some $k$ then the inequality $\sem{\cA}(w(S, i_0 + i)) \le \sem{\cA}(w(S, i_0 + i+1))$ would be sharp, which is a contradiction.
For the other direction suppose there exists a run $\rho \in
\Run_{\cA}(w)$ such that $c_{\bar{\rho}}^k = 0$ for every $k \in
S$. Then for every $i \ge 0$ there exists a run $\rho_i \in
\Run_{\cA}(w(S, i_0 + i))$ such that $|\rho_i| \le |\rho| +
\sum_{k}d_{\bar{\rho}}^k$. Since $\sem{\cA}(w(S, i_0 + i)) \le
\sem{\cA}(w(S, i_0 + i+1)) \le |\rho| + \sum_{k}d_{\bar{\rho}}^k$
it follows that $\sem{\cA}(w(S, i_0 + i)) = \sem{\cA}(w(S, i_0 +
i+1))$ for all sufficiently large $i$.

Given the previous discussion, let $\bar{R}_k = \{\bar{\rho} \in \overline{\Run}_\cA(w) \mid c_{\bar{\rho}}^k > 0\}$ for every $k \in \{1, \ldots, n\}$.
The set $\bar{R}_k$ represents the abstractions of the runs over $w$ that will grow when pumping $w(\{k\}, i)$.
Then, we can restate~\eqref{eq:nequal} as: $\sem{\cA}(w(S, i)) <
\sem{\cA}(w(S, i+1))$ for all sufficiently large $i$ iff $\bigcup_{k \in S} \bar{R}_k = \overline{\Run}_\cA(w)$.

We are now ready to prove the theorem.
Fix a partition $S_1, \ldots, S_m$ of $\{1,\ldots,n\}$ for some $m \geq \varphi(\max
|S_l|)$ ($\varphi$ will be defined below).
Suppose the first condition is not true, namely, for all $1 \leq j
\leq m$ there
exist arbitrarily big values $i$ such that $f(w(S_j, i)) \neq
f(w(S_j,i+1))$. From~\eqref{eq:nequal} it follows that $f(w(S_j, i))
< f(w(S_j,i+1))$ for all sufficiently large $i$ and $\bigcup_{k \in S_j} \bar{R}_k = \overline{\Run}_\cA(w)$ for every $j \leq m$.
Let $L = \max |S_l|$. We assume that $L > 1$, otherwise every
selection $S$ contains a whole set $S_k$ for some $k$ and we are
done.

To construct the selection set $S = \{k_1, \ldots, k_m\}$ we define by induction the sets $G_j$.
For every $j \in \{1, \ldots,
m\}$ let $G_j = \overline{\Run}_\cA(w) \setminus \bigcup_{l \leq j}
\bar{R}_{k_l}$ (where $k_0$ is undefined, so $G_0 =
\overline{\Run}_\cA(w)$). Intuitively, $G_j$ correspond to runs that are not covered by the set $\{k_1, \ldots, k_j\}$.
For the inductive case, suppose that $j \geq 0$ and $G_j \neq \emptyset$.
Since $\bigcup_{k \in S_{j+1}} \bar{R}_k = \overline{\Run}_\cA(w)$,
by the pigeonhole principle there exist $k_{j+1} \in S_{j+1}$ such that $|\bar{R}_{k_{j+1}} \cap G_j| \geq |G_j|/|S_{j+1}|$. We add $k_{j+1}$ to $S$ and so
$|G_{j+1}| \leq |G_j| - |G_j|/|S_{j+1}| = |G_j|\cdot (|S_{j+1}| -1 )/|S_{j+1}| \leq |G_j|\cdot (L -1)/L$.
Suppose this procedure continues until $j = m$ and $G_{m} \neq \emptyset$. Then $1 \le |\overline{\Run}_\cA(w)|\cdot ((L-1)/L)^m$, and $|\overline{\Run}_\cA(w)| \ge (L/(L-1))^m$. However, we know that $|\overline{\Run}_\cA(w)|$ is bounded by a polynomial function $P(n)$ depending on $|\cA|$. Thus, it suffices to choose  $\varphi$ such that $m \geq \varphi(L)$ implies $(L/(L-1))^{m} > P(L\cdot m) \geq P(n) \ge |\overline{\Run}_\cA(w))|$ (recall that $S_1, \ldots, S_m$ is a partition of $\{1, \ldots, n\}$ and $L \cdot m \geq n$). Therefore, $G_m = \emptyset$ and thus $\bigcup_{k \in S} \bar{R}_k = \overline{\Run}_\cA(w)$, which concludes the proof.
\end{proof}

\section{Pumping lemmas for the max-plus semiring}
\label{sec:max}
In this section, we consider finitely ambiguous and polynomially ambiguous weighted automata over the $\arctic$ semiring.
Notice that U-WA over $\arctic$ is the same class of functions as U-WA over $\trop$ and thus Theorem~\ref{theorem:pumping_refined} also holds for this class.
For this reason, here we focus on the ambiguous cases, dividing the section into two parts to deal separately with the finitely ambiguous and polynomially ambiguous cases.

\subsection{Pumping Lemma for Finitely Ambiguous Weighted Automata over $\arctic$}

We use the definitions of \textit{refinements} and \textit{n-pumping
  representations} from Section~\ref{sec:pump-sum} and
Section~\ref{sec:fin_min}. In order to formulate the pumping lemma for
finitely-ambiguous functions over the $\mathbb{N}_{\max, +}$ semiring,
we define a few more notations.

Fix a function $f : \Sigma^* \to \mathbb{N}$ and a word $w \in
\Sigma^*$. Suppose that we have an $n$-pumping representation for $w$,
and a refinement thereof. Let $\{1,\ldots,n\}$ be the set of all
indices in the refinement. We say that a refinement is
\emph{\proper} if for every subset $S \subseteq \{1,\ldots,n\}$, there exists $K$
such that $f(w(S, i + 1)) = K + f(w(S, i))$ for all sufficiently
large $i$.
For \proper \ refinements we let $\Delta(S)$ denote the above value
$K$ (note that $\Delta$ depends on $f$ and $w$, which are
fixed). Furthermore, we say that $S \subseteq \{1,\ldots,n\}$ is \textit{\compatible} if
\[
\Delta(S) = \sum_{j \in S} \Delta(\{j\}).
\]

\begin{thm}\label{thm:fin_max_plus}
  Let $f : \Sigma^* \to \mathbb{N}$ be a finitely ambiguous function
  over the  semiring $\mathbb{N}_{\max, +}$. There exists $N \in
  \mathbb{N}$ such that for every $n$-pumping representation
  \[
  w = u_0 \cdot \underline{v_1} \cdot u_1 \cdot \underline{v_2} \cdot \ldots \underline{v_{n}} \cdot u_{n},
  \]
  where $n \ge N$ and and $|v_i| \ge N$ for all $i$, there exists a \proper \ refinement
  \[
  w = x_0 \cdot \underline{y_1} \cdot x_1 \cdot \underline{y_2} \cdot \ldots \underline{y_{n}} \cdot x_{n}
  \]
  such that for every sequence of pairwise different, non-empty sets
  $S_1, S_2, \ldots S_{k} \subseteq \{1,\ldots,n\}$ with $k \ge N$,
  one of the following holds:
  \begin{itemize}
  \item there exists $j$ such that $S_j$ is not \compatible;
  \item there exist $j_1$ and $j_2$ such that $\{l_1, l_2\}$ is decomposable for every $l_1 \in S_{j_1}$ and $l_2 \in S_{j_2}$.
  \end{itemize}
\end{thm}

Before proving the theorem we show how to use the pumping lemma on two
examples.

\begin{exa}\label{ex:max-a_to_b}
  Consider the function $g_3$ which computes
  \[
  \max_{0 \leq i \leq n} {|a_1 \ldots a_i|_a + |a_{i+1} \ldots
    a_n|_b},
  \]
  for any $w=a_1 \dots a_n  \in \{a , b\}^*$. This is defined by a small modification of $\mathcal{W}_3$ in
  Figure~\ref{fig:WA}. The weights of the transitions changing the states is modified to $1$ and the semiring is changed from
  $\trop$ to $\arctic$. We show that this function cannot be expressed
  by any finitely ambiguous WA over $\arctic$. Towards a contradiction
  fix $N$ from Theorem~\ref{thm:fin_max_plus} and consider the
  $(2N+2)$-pumping representation
  $(\underline{a^{N+1}}\,\underline{b^{N+1}})^{N+1}$. In the
  refinement, we index the $j$-th block of $a$'s with $j$ and the
  $j$-th block of $b$'s with $j'$. Let $x_1,\ldots,x_{N+1}$ and
  $y_1,\ldots,y_{N+1}$ be the lengths of all blocks of $a$'s and $b$'s
  in the refinement.
  We define the sets $S_j = \{j, j'\}$ for all $1 \leq j \leq N + 1$. We show that none of the conditions of the pumping lemma hold. First, 
  it is easy to see that all sets $S_j$ are \compatible. Indeed $\Delta(S_j) = x_j + y_j = \Delta(\{j\}) + \Delta(\{j'\})$.
  For the second condition we can assume that $j_2 > j_1$. Since the
  function counts $a$'s before $b$'s, the set $\{j_1', j_2\}$ is not
  decomposable because $\Delta(\{j_1',j_2\}) =
  \max(y_{j_1},x_{j_2})$. Thus, any $S_{j_1}$ and $S_{j_2}$ will not
  satisfy the second condition either.
\end{exa}

\begin{exa}
  Consider the function $g_4$ which computes the length of longest block of $b$'s. This is defined by $\mathcal{W}_4$ in Figure~\ref{fig:WA} if we change the semiring of the automaton from $\trop$ to $\arctic$. We shall show that $g_4$ cannot be expressed by a FA-WA over $\arctic$. Towards a contradiction let $N$ be the constant from Theorem~\ref{thm:fin_max_plus}. Consider the $(N+1)$-pumping representation $(\underline{b^{N+1}}a)^{N+1}$. We define the sets $S_j = \{j\}$ for all $1 \leq j \leq N + 1$.
  First, each set is also \compatible \ for trivial reasons. Second, every index is not decomposable with any other since the function takes into account the value of at most one block of $b$'s.
 \end{exa}

\begin{proof}[Proof of Theorem~\ref{thm:fin_max_plus}]
  Let $\mathcal{A}$ be the finitely ambiguous WA that
  computes $f$. Suppose that $\mathcal{A}$ has ambiguity at most $m$.
  Let $N = \max\{K, m+1\}$ such that $K$ is the constant from Lemma~\ref{lemma:ramsey} applied to $\bool^{Q \times Q}$.
  Now, consider a word $w$ and an $n$-pumping representation of $w$
  according to the theorem. Since each $v_j$ in the representation has
  length more than $K$ we can refine every $v_j$ to $y_j$ such that every $\bar{M}_{y_j}$ is an idempotent.

  We prove that this refinement is \proper. Fix $S \subseteq \{1,\ldots,n\}$. We prove that $f(w(S, i + 1)) = K' + f(w(S, i))$ such that $K'$ does not depend on $i$ for all $i$ big enough. Consider $i > |Q|$.
  Let $\rho$ be an accepting run over $w(S, i)$. For every $j \in S$ let $p_{j,0}$ be the state in the run $\rho$ before reading $y_j$ and $p_{j,1}, \ldots p_{j,i}$ the states after reading $y_j$ the respective amount of times. Notice that by definition $\bar{M}_{y_j}(p_{j,s},p_{j,s+1}) = 1$ for all $s \in \{0,\ldots,i-1\}$. We will need some observations that follow from the assumption that $\cA$ is finitely-ambiguous.

  First, let $0 \le s < s' \le i$ such that $p_{j,s} = p_{j,s'}$. We prove that $p_{j,s} = p_{j,s''}$ for all $s \le s'' \le s'$. Since $\rho$ is an accepting run and $\bar{M}_{y_j}$ is idempotent $\bar{M}_{y_j}(p_{j,s},p_{j,s''}) = \bar{M}_{y_j}(p_{j,s''},p_{j,s}) = 1$ for all $s \le s'' \le s'$. Therefore, if there exists $s''$ such that $p_{j,s} \neq p_{j,s''}$ then there would be two different accepting runs from $p_{j,s}$ to $p_{j,s}$ when reading $(y_j)^2$ (with $p_{j,s}$ or $p_{j,s''}$ in the middle). This contradicts the EDA criterion in~\cite{weber1991degree} (by pumping $(y_j)^2$ one could generate arbitrary many accepting runs).

  To state the second property we introduce some notation. We say that $(s,s')$ is a maximal cycle in $p_{j,0},\ldots,p_{j,i}$ if $0 \le s < s' \le i$, $p_{j,s} = p_{j,s'}$ and $p_{j,s-1} \neq p_{j,s}$, $p_{j,s'} \neq p_{j,s'+1}$ (the last condition is required for $s > 0$ and $s' < i$). We prove that there is at most one maximal cycle in every $p_{j,0},\ldots,p_{j,i}$. Suppose otherwise that there are $0 \le s < s' < t < t' \le i$ such that $(s,s')$ and $(t,t')$ are maximal cycles. Since $\rho$ is an accepting run and $\bar{M}_{y_j}$ is idempotent $\bar{M}_{y_j}(p_{j,s},p_{j,s}) = \bar{M}_{y_j}(p_{j,s},p_{j,t}) = \bar{M}_{y_j}(p_{j,t},p_{j,t}) = 1$. This contradicts the IDA criterion in~\cite{weber1991degree}, i.e.\ for every $k$ when reading when reading $(y_j)^k$ there would be at least $k$ runs from $s$ to $t$.

  We conclude that every run $\rho$ over $w(S, i)$ has a unique maximal cycle in every block $j \in \{1,\ldots,n\}$. Notice that the number of accepting runs over $w(S, i+1)$ can only increase compared to the number of runs over $w(S, i)$ since every accepting run $\rho$ over $w(S, i)$ can be prolonged by increasing the maximal cycle. Recall that $\cA$ is finitely-ambiguous. Therefore, for $i$ big enough the number of accepting runs over $w(S, i)$ and $w(S, i+1)$ will be always the same. We denote this number by $m$.

%

   Let $\rho_1,\ldots,\rho_m$ be all accepting runs over $w(S,i)$. Let $\wt[\rho_s]$ denote the weight of the run $\rho_s$ for $1\le s \le m$, i.e.\ $f(w(S,i)) = \max(\wt(\rho_1),\ldots,\wt(\rho_m))$. For every $j \in \{1,\ldots,n\}$ and $l \in \{1,\ldots,m\}$ we denote by $\wt(\rho_l[j])$ the weight ${M}_{y_j}(p,p)$, where $p$ is the state defining the corresponding maximal cycle.
   Note that by construction, each $\rho_l[j]$ is a cycle.
   Then
 \begin{align*}
     f(w(S,i+1)) = \max_{l \in \{1,\ldots,m\}} \wt(\rho_l) + \sum_{i \in S} \wt{\rho_l[i]}.
 \end{align*}
   Therefore this refinement is \proper, where $\Delta(S) = \max_{l \in \{1,\ldots,m\}} \sum_{i \in S} \wt{\rho_l[i]}$.

   We say that a cycle
   $\rho_l[j]$ is dominant if $\wt(\rho_{l'}[j]) \leq
   \wt(\rho_{l}[j])$ for all $l'$.
  The rest of the proof involves reasoning about the dominant
  cycles. First we make a simple observation. Suppose $\rho_l[j]$ is
  dominant, then $\wt(\rho_l[j]) = \Delta(\{j\})$. We make one more observation.

  \begin{clm}\label{clm}
    Assume that we have a \proper\ refinement and $S \subseteq
    \{1,\ldots,n\}$ is a subset of indices of the refinement. Then $S$
    is \compatible\ if and only if there exists some run $\rho$ such that for all $j \in S$, $\rho[j]$ is dominant.
  \end{clm}

  \begin{claimproof}
    Assume first that $S$ is \compatible, and consider some run $\rho$ such that $\sum_{j \in
      S}\wt(\rho[j])$ is maximal among all runs on $w$. We claim that
     $\rho[j]$ is dominant for all $j \in S$. Assume this is not the
     case. Notice that the value computed by the automaton increases
     by $\sum_{j \in S}\wt(\rho[j])$ when $S$ is pumped in $\rho$. By
     the choice of $\rho$ and the fact that the refinement is \proper,
     we have that $\Delta(S) =  \sum_{j \in S}\wt(\rho[j])$.  By
     assumption we know that there is some $j^* \in S$ such that
     $\rho$ is not dominant for $j^*$, so  $\wt(\rho[j^*]) <
     \Delta(\{j^*\})$. But this means that $\Delta(S) = \sum_{j \in
       S}\wt(\rho[j]) < \sum_{j \in S}{\Delta(\{j\})}$, which is a
     contradiction to $S$ being \compatible.

     For the reverse implication, consider a run $\rho$ such that $\rho[j]$ is dominant for all $j \in S$.
     In particular, $\sum_{j \in S}\wt(\rho[j]) \geq
     \sum_{j\in S} \wt(\rho'[j])$ for any other run $\rho'$. This
     means that when the set $S$ is pumped, the value computed by the
      automaton increases by $\sum_{j \in S}\wt(\rho[j])$, which also happens to be $\sum_{j \in S}\Delta(\{j\})$ since the cycles in consideration are dominant.
   \end{claimproof}

  To conclude we show that if the first condition of the pumping lemma does not hold then the second condition must hold. Indeed, suppose that all sets are \compatible. Then by Claim~\ref{clm} for every set $S_j$ there is some run $\rho_{l_j}$ in which all the cycles corresponding to $S_j$ are dominant. But since there are more sets than runs, there must be some $j_1 \neq j_2$ such that $l_{j_1} = l_{j_2}$, namely, two sets which have the same corresponding runs. However, by Claim~\ref{clm} this means that $\{k_1, k_2\}$ is decomposable for every $k_1 \in S_{j_1}$ and $k_2 \in S_{j_2}$.
\end{proof}

\subsection{Pumping Lemma for Polynomially Ambiguous Weighted Automata over $\arctic$}

In this section, we will re-use the definition of linear refinement
and decomposability from the previous section. We will also re-use the definition of  selection set from Section~\ref{sec:pumping}.

\begin{thm}\label{thm:poly_max_plus}
Let $f : \Sigma^* \to \mathbb{N}$ be a polynomial-ambiguously recognisable function over $\mathbb{N}_{\max, +}$. There exist $N$ and a function $\varphi : \mathbb{N} \to \mathbb{N}$ such that for all $n$-pumping representations
$$
w = u_0 \cdot \underline{v_1} \cdot u_1 \cdot \underline{v_2} \cdot \ldots u_{n-1} \cdot \underline{v_n} \cdot u_n,
$$
where $|v_i| \geq N$ for every $1 \le i \leq n$, there exists a \proper \ refinement
$$
w = u_0' \cdot \underline{y_1} \cdot u_1' \cdot \underline{y_2} \cdots u_{n-1}' \cdot \underline{y_n} \cdot u_n',
$$
such that for every partition $\pi = S_1, S_2, \ldots S_m$ of $\{1, \ldots , n\}$ with $m \ge \varphi(\max_j(|S_j|))$  one of the following holds:
\begin{itemize}
\item there exists $j$ such that $S_j$ is decomposable;
\item there exists a selection set $S$ for $\pi$ such that $S$ is not decomposable.
\end{itemize}
\end{thm}

Before proving this lemma, we show how to use it on examples.

\begin{exa}\label{ex:polymax1}
Consider the function $g_5$ such that, for any $w$ of the form $w_0 \# w_1 \# \ldots \# w_n$ with $w_i \in \{a,b\}^*$ it computes $g_5(w) = \sum_{i=0}^n \max\{|w_i|_a, |w_i|_b\}$.
This is defined by $\mathcal{W}_5$ in Figure~\ref{fig:WA} if we change the semiring of the automaton to $\arctic$.
We show that $g_5$ cannot be expressed by a PA-WA.
Assume the contrary and let $N$ and $\varphi$ be the constant and the function from Theorem~\ref{thm:poly_max_plus}. Let $m$ be a number larger than $\varphi(2)$. We consider the refinements of the pumping representation $(\underline{a^N}\, \underline{b^N} \#)^m$. In the refinement, we refer to the $j$-th block of $a$'s as $j$ and to the $j$-th block of $b$'s as $j'$.
We define the sets in the partition as $S_j = \{j, j'\}$ for all $1
\leq j \leq m$. It is clear from the definition of the function $g_5$
that no set $S_j$ is decomposable, since only one of the blocks $j$
and $j'$ is relevant for the outer sum. Now, consider any selection
set $S$. Since any two elements of $S$  belong to different blocks of
the word separated by $\#$'s, $S$ is a decomposable set of
indices. Therefore,  the function $g_5$ is not polynomially ambiguous
over $\arctic$.
\end{exa}

\begin{exa}\label{ex:polymax2}
Consider the function $g_6$ that given a word of the form $w_0 \# w_1
\# \ldots \# w_n$ with $w_i \in \{a,b\}^*$  computes $g_6(w) =
\sum_{i=0}^n g_3(w_i)$, where $g_3$ from
Example~\ref{ex:max-a_to_b}.
We show now that $g_6$ cannot be expressed as a PA-WA over $\arctic$.
Assume the contrary and let $N$ and $\varphi$ be the constant and the function from the lemma above. Let $m$ be a number larger than $\varphi(2)$. We consider the refinements of the pumping representation $(\underline{b^N}\, \underline{a^N} \#)^m$.
Like in Example~\ref{ex:polymax1} in the refinement we refer to the $j$-th block of $b$'s as $j$ and to the $j$-th block of $a$'s as $j'$. Let $x_j$ and $y_j$ be the lengths of the block of $b$'s and the block of $a$'s, respectively.
We define the sets in the partition as $S_j = \{j, j'\}$ for all $1 \leq j \leq m$.
Every set $S_j$ is not decomposable since $\Delta(S_j) = \max(x_j,y_j)$.
Consider any selection set $S$ of $\pi$. It is easy to see that $S$ is decomposable given that all elements belong to different blocks. We conclude that $g_6$ cannot be defined by PA-WA over $\arctic$.
\end{exa}

\begin{proof}[Proof of Theorem~\ref{thm:poly_max_plus}]
The first part of the proof is the same as in the proof of Theorem~\ref{theorem:polynomial} which only depends on the (polynomial) ambiguity of the automaton, and not on the semiring.
We will use the same same notation for the refinement  $(y_k)_k$ and
in particular, all $D_k = M_{y_k}$ are idempotents.
We will also use the notation $\overline{\Run_\cA}(w)$ of abstractions of runs and the notation $\bar{\rho}: \{1, \ldots, n\} \to Q \times Q$.
Finally recall from the proof of Theorem~\ref{theorem:polynomial} that $|\overline{\Run_\cA}(w)| \leq P(n)$ for some polynomial $P(\cdot)$.

We will also reuse Lemma~\ref{lemma:polymatrix}, but over the
$\arctic$ semiring. One can easily check that this lemma continues to hold over the max-plus semiring. The difference is that then $c,d \in \arctic$.
Therefore, for every $k \in S$ let $b_{\bar{\rho}(k)}^k$, $c_{\bar{\rho}(k)}^k$ and $d_{\bar{\rho}(k)}^k$ be the constants from Lemma~\ref{lemma:polymatrix} such that:
$$
D_k^{b_{\bar{\rho}(k)}^k + i}[\bar{\rho}(k)] = c_{\bar{\rho}(k)}^k \cdot i + d_{\bar{\rho}(k)}^k
$$
 for $i \ge 0$. Since $\rho$ is accepting, we have $c_{\bar{\rho}(k)}^k, d_{\bar{\rho}(k)}^k \neq - \infty$.

First, we argue that the refinement defined by $(y_k)_k$ is \proper.
Fix a non-empty set $S \subseteq \{1, \ldots, n\}$.
For $\bar{\rho} \in \overline{\Run_\cA}(w)$ let $c_{\bar{\rho}} =
\sum_{k \in S} c_{\bar{\rho}(k)}^k$. Recall that every run in
$\Run_\cA(w(S,i))$ has some abstraction in $\overline{\Run_\cA}(w)$
and $\cA$ outputs the maximal value among all runs. It follows by
considering $i$ big enough that $\Delta(S) = \max\{c_{\bar{\rho}} \mid \bar{\rho} \in \overline{\Run_\cA}(w)\}$.

Let $k \in \{1, \ldots, n\}$. We say that $\overline{\rho} \in \overline{\Run_{\cA}(w)}$ is $k$-dominant if $c_{\bar{\rho}(k)}^k \ge c_{\bar{\sigma}(k)}^k$ for every $\bar{\sigma} \in \overline{\Run_\cA}(w)$. 
%
Given $k \in \{1, \ldots, n\}$ we define $\bar{R}_k =
\{\overline{\rho} \in \overline{\Run_\cA}(w) \mid \overline{\rho}
\text{ is not $k$-dominant}\}$.

\begin{clm}\label{clm2}
    Assume that we have a \proper\ refinement and $S \subseteq
    \{1,\ldots,n\}$ is a subset of indices of the refinement. Then $S$
    is \compatible\ if and only if there exists $\overline{\rho}$ such that for all $j \in S$, $\overline{\rho}[j]$ is $j$-dominant.
\end{clm}

\begin{claimproof}
Follows the same steps as the proof of Claim~\ref{clm}.
\end{claimproof}

\medskip

\noindent By Claim~\ref{clm2}, $S$ is not decomposable if and only if $\bigcup_{k \in S} \bar{R}_k = \overline{\Run_\cA}(w)$.

\medskip

We are ready to prove the theorem.
Fix a partition $S_1, \ldots, S_m$ for some $m \geq \varphi(\max_l |S_l|)$.
Suppose the first condition is not true, namely, for every $j$, the set $S_j$ is not decomposable.
Let $L = \max_l |S_l|$. Since no set $S_l$ is decomposable we know that $L > 1$.
By the observations in the previous paragraph it suffices to construct a selection set $S$ such that $\bigcup_{k \in S} \bar{R}_k = \overline{\Run_\cA}(w)$, which will imply that $S$ is not decomposable.

The remaining part of the proof follows the same steps as the last part in the proof of Theorem~\ref{theorem:polynomial}.
To construct the selection set $S = \{k_1, \ldots, k_m\}$ we define by induction the sets $G_j$.
For every $j \in \{1, \ldots, m\}$ let $G_j = \overline{\Run}_\cA(w)
\setminus \bigcup_{l \leq j} \bar{R}_{k_l}$ (where $k_0$ is undefined,
so that $G_0 =
\overline{\Run}_\cA(w)$). Intuitively, $G_j$ correspond to runs that are not covered by the set $\{k_1, \ldots, k_j\}$.
For the inductive case, suppose that $j \geq 0$ and $G_j \neq \emptyset$.
Since $\bigcup_{k \in S_{j+1}} \bar{R}_k = \overline{\Run}_\cA(w)$,
by the pigeonhole principle there exist $k_{j+1} \in S_{j+1}$ such that $|\bar{R}_{k_{j+1}} \cap G_j| \geq |G_j|/|S_{j+1}|$. We add $k_{j+1}$ to $S$ and so
$|G_{j+1}| \leq |G_j| - |G_j|/|S_{j+1}| = |G_j|\cdot (|S_{j+1}| -1 )/|S_{j+1}| \leq |G_j|\cdot (L -1)/L$.
Suppose this procedure continues until $j = m$ and $G_{m} \neq \emptyset$. Then $1 \le |\overline{\Run}_\cA(w))|\cdot ((L-1)/L)^m$, and $|\overline{\Run}_\cA(w)| \ge (L/(L-1))^m$. However, we know that $|\overline{\Run}_\cA(w)|$ is bounded by a polynomial function $P(n)$ depending on $|\cA|$. Thus, it suffices to choose  $\varphi$ such that $m \geq \varphi(L)$ implies $(L/(L-1))^{m} > P(L\cdot m) \geq P(n) \ge |\overline{\Run}_\cA(w)|$ (recall that $S_1, \ldots, S_m$ is a partition of $\{1, \ldots, n\}$ and $L \cdot m \geq n$). Therefore, $G_m = \emptyset$ and thus $\bigcup_{k \in S} \bar{R}_k = \overline{\Run}_\cA(w)$, which concludes the proof.
\end{proof}

\section{Conclusion}
\label{sec:conclusions}
We have shown five pumping lemmas for five different classes of functions. We believe that the pumping lemmas in Section~\ref{sec:pumping} and in Section~\ref{sec:max} could be proved for a wider class of functions that would contain the class $\natsemiring$, but this is left for future work.
As a corollary of our results, we showed that recognisable functions over $\trop$ and $\arctic$ form a strict hierarchy, namely:
$$
\text{U-WA} \ \subsetneq \ \text{FA-WA} \ \subsetneq \ \text{PA-WA} \ \subsetneq \ \text{WA}.
$$
All strict inclusions, except for PA-WA $\subsetneq$ WA, could be extracted from the analysis of examples in~\cite{KlimannLMP04}. However, our results provide a general machinery to prove such results.

\paragraph*{Acknowledgments}
We thank Shaull Almagor and the anonymous referees of STACS 2018 and LMCS for their helpful comments.


\bibliographystyle{alpha}
\bibliography{biblio}

%


\end{document}